\documentclass[journal]{IEEEtran}

  \usepackage{cite}

\usepackage{bm}
\usepackage{amsmath}
\usepackage{epsf}
\usepackage{import}
\usepackage{graphics}
\usepackage{ amssymb }
\usepackage[dvips]{graphicx}
\usepackage{epsfig}
\usepackage{cite}
\usepackage[linesnumbered,ruled,vlined]{algorithm2e}
\usepackage{colortbl}
\usepackage{color}

\usepackage{bm}
\usepackage{amsmath}
\usepackage{epsf}
\usepackage{graphics}
\usepackage{ amssymb }
\usepackage[dvips]{graphicx}
\usepackage{epsfig}
\usepackage{cite}
\usepackage[linesnumbered,ruled,vlined]{algorithm2e}
\usepackage{graphicx}
\usepackage{epsfig}
\usepackage{latexsym}
\usepackage{amsfonts}
\usepackage{here}
\usepackage{rawfonts}
\usepackage{hyperref}

\usepackage[utf8]{inputenc}
\usepackage[english]{babel}
\usepackage{amsmath}
\usepackage{amsfonts}
\usepackage{amssymb}
\usepackage{hyperref}
\usepackage{bm}
\usepackage{listings}
\usepackage{caption}
\usepackage{amssymb}
\usepackage{amsthm}
\usepackage{graphicx}
\usepackage{epstopdf}
\usepackage{listings}
\usepackage{float}
\usepackage{amsmath}
\usepackage{amssymb}
\usepackage{amsfonts}
\usepackage{epstopdf}
\usepackage[overload]{empheq}

\usepackage{multirow}
\usepackage{amscd}
\usepackage{mathrsfs}
\usepackage{graphicx}
\usepackage{color}
\usepackage{url}
\usepackage{bm}
\usepackage{setspace}
\usepackage{footnote}
\usepackage{enumitem}
\DeclareMathOperator*{\argmin}{arg\,min}
\newtheorem{theorem}{Theorem}
\newtheorem{lemma}{Lemma}
\newtheorem{definition}{Definition}
\newtheorem{proposition}{Proposition}

\newtheorem{remark}{Remark}

\newtheorem{fact}{Fact}

\usepackage[]{algorithm2e}

\renewcommand{\figurename}{Fig.}
\addto\captionsenglish{\renewcommand{\figurename}{Fig.}}

\newcommand\norm[1]{\left\lVert#1\right\rVert}
\usepackage{mathtools}

\usepackage{lipsum,graphicx,subcaption}

\usepackage{diagbox}
\usepackage[protrusion=true,expansion=true]{microtype}
\usepackage[font=small]{caption}
\usepackage[linesnumbered]{algorithm2e}
\makeatletter
\newcommand{\nosemic}{\renewcommand{\@endalgocfline}{\relax}}
\newcommand{\dosemic}{\renewcommand{\@endalgocfline}{\algocf@endline}}
\let\oldnl\nl
\newcommand{\nonl}{\renewcommand{\nl}{\let\nl\oldnl}}
\newcommand\semiHuge{\fontsize{20.72}{27.38}\selectfont}

\begin{document}

\title{\semiHuge{Energy-Aware Stochastic UAV-Assisted Surveillance}}

	\author{Seyyedali Hosseinalipour,~\IEEEmembership{Member,~IEEE,} Ali Rahmati, \IEEEmembership{Student~Member,~IEEE,} Do~Young~Eun,~\IEEEmembership{Senior~Member,~IEEE,} and~Huaiyu~Dai,~\IEEEmembership{Fellow,~IEEE}
	\IEEEcompsocitemizethanks{\IEEEcompsocthanksitem{S. Hosseinalipour is with School of Electrical and Computer Engineering, Purdue University, West Lafayette, IN, USA. email:\{hosseina@purdue.edu\}. This work was completed while he was with the Department
of Electrical and Computer Engineering, North Carolina State University, Raleigh,
NC, USA.}\IEEEcompsocthanksitem{A. Rahmati, D. Y. Eun, and H. Dai are with the Department
of Electrical and Computer Engineering, North Carolina State University, Raleigh,
NC, USA. e-mail: \{arahmat,dyeun,hdai}@ncsu.edu\}.
\IEEEcompsocthanksitem{This work was supported by the US National Science Foundation under grant CNS-1824518.}
}}
	

	\maketitle


	\begin{abstract}
With the ease of deployment, capabilities of evading the jammers and obscuring their existence,	unmanned aerial vehicles (UAVs) are one of the most suitable candidates to perform surveillance. There exists a body of literature in which the inspectors follow a deterministic trajectory to conduct surveillance, which results in a predictable environment for malicious  entities. Thus, introducing randomness to the surveillance is of particular interest. In this work, we propose a novel framework for stochastic UAV-assisted surveillance that i) inherently considers the battery constraints of the UAVs, ii) proposes random moving patterns modeled via random walks, and iii) adds another degree of randomness to the system via considering probabilistic inspections. We formulate the problem of interest, i.e., obtaining the energy-efficient random walk and inspection policies of the UAVs subject to probabilistic constraints on inspection criteria of the sites and battery consumption of the UAVs, which turns out to be signomial programming that is highly non-convex. To solve it, we propose a centralized and a distributed algorithm along with their performance guarantee. This work contributes to both UAV-assisted surveillance and classic random walk literature by designing random walks with random inspection policies on weighted graphs with energy limited random walkers.
	\end{abstract}
  \vspace{-1mm}

	\IEEEpeerreviewmaketitle


	
\section{Introduction}

  \noindent \IEEEPARstart{R}{ecently},  unmanned aerial vehicles (UAVs) have attracted lots of attention due to their low cost and flexibility of deployment~\cite{zeng2019accessing,gupta2015survey,fotouhi2019survey,hosseinalipour2020federated}. With the recent advances, UAVs are assumed as a  promising alternative for ground robots to conduct surveillance in various applications~\cite{gu2018multiple}.  As compared to surveillance conducted using ground vehicles or robots, the UAV surveillance assumes the following advantages: (i) the UAVs can move at higher speeds, which results in shorter turnaround times between the targeted sites; (ii) the UAVs can fly at higher altitudes, making them less susceptible to interference and attacks from terrestrial entities; (iii) deployment of UAVs is more flexible (requiring no preexisting roads and road construction) and appealing for hard-to-reach areas (e.g., dense forests and seas/oceans); (iv) UAVs possess high agility and can instantly change their trajectories upon request; (v) UAVs can be deployed over long distances in autonomous manner with low risk of physical accidents (e.g., road hazard and collision in case of autonomous ground vehicles).
  
  The UAV-assisted surveillance  includes collecting information, typically images or videos, about specific targets. 
    In recent literature, the UAV is
utilized as a flying camera over a given area following a predefined deterministic trajectory, which is already optimized with respect to the given network constraints. In particular, such systems consist of
several checkpoints that UAVs are required to visit in their routes to the target location~\cite{IntroDet1, IntroDet2}. 
However, having deterministic trajectories for the UAVs leads to multiple security- and privacy-related concerns. In particular, the malicious entities can  predict the exact locations of UAVs and also the visiting/inspection times of their areas of interest, using which they can achieve their goals while remaining unnoticed~\cite{beard2006decentralized  , yue2018software,leu2019survivable}. Thus, it is of high importance and interest to leverage stochastic movement in surveillance. Another key factor to consider is that UAVs' limited battery capacity does not allow them to fly for an unlimited period of time~\cite{motlagh2017uav}. Hence, the designed trajectories should be both unpredictable and energy efficient. 

   \subsection{Related Works}\label{relatedWork}
 In a large body of existing UAV-assisted literature, the trajectory/location of the UAVs are optimized subject to network constraints, where the UAVs are considered as relaying nodes forwarding data. Considering multiple static UAVs, optimal UAV locations are derived in \cite{8424236} through maximizing the data rate. We studied the optimal position planning of UAV relays considering the effect of interference in the environment \cite{hosseinalipour2019interference, hosseinalipour2019interference1}.  A UAV-assisted communication scheme is proposed in~\cite{zhang2018joint}, where the UAV trajectory, and the transmit power of both the UAV and the mobile device are obtained to minimize
the outage probability. Taking  advantage of the inherent mobility feature of the UAVs,  an adaptive interference avoidance position planning scheme is developed in~\cite{rahmati2019interference}. Furthermore, trajectory design and path planning is studied with respect to power control in multi-UAV systems~\cite{wu2018joint}, minimizing the energy for wireless transfer~(WPT)-enabled UAVs~\cite{xu2018uav}, and search and localization \cite{tisdale2009autonomous}. 
Relevant works concerning the usage of UAVs in surveillance application include multi-UAV surveillance in complex urban
environments with occlusions \cite{semsch2009autonomous}, low cost vision-based indoor UAV autonomous patrolling \cite{lee2015autonomous}, and cooperative perimeter surveillance \cite{kingston2008decentralized}. Moreover, in  
 \cite{li2019energy}, an energy efficient UAV surveillance scenario is investigated, where a
 proactive eavesdropping scheme  is proposed
to facilitate the eavesdropping and jamming for
the legitimate UAV  to maximize the amount of packets eavesdropped from the suspicious UAVs' communication. 
In \cite{8754781}, the authors aim to maximize the average surveillance
 rate by optimizing the position and the jamming power of the legitimate UAV monitor. 
There exists a body of literature on robotic patrolling including  a realistic model of robot motion control with velocity uncertainties~\cite{elmaliach2008realistic}, reinforcement learning to achieve efficient cooperative behavior among the agents~\cite{santana2004multi}, and monitoring the locations of interest~\cite{smith2010multi}.
In  \cite{pasqualetti2012cooperative},  the concept of refresh time
and latency of a team trajectory is introduced, where a procedure is proposed to build
a road-map to represent the topological structure of the
area to be patrolled. In \cite{sak2008probabilistic},   the problem of obtaining the visiting sequence of the nodes for multiple homogeneous agents is studied considering three different types of intruders. The majority of this literature focuses on studying the traveling salesman problem in different contexts. The main difference between our work and all the mentioned works lies in designing stochastic yet energy-efficient movement and inspection policies, which results in an unpredictable and secure surveillance design.
 Finally, a body of literature is devoted to studying different applications of random walks including community detection, spectral algorithms for independent sampling \cite{spitzer2013principles, lambiotte2014random, gkantsidis2004random}, and data gathering~\cite{huang2018cost}. 
\subsection{Novelty and Contributions}
In this paper, we propose a new model for energy-aware stochastic UAV-assisted surveillance, leveraging random walks and stochastic inspection policies. The goal is to design an optimal stochastic kernel for random walks and the inspection policies of multiple UAVs with limited battery capacities to minimize the long-term average of their energy consumption,  while addressing the constraints on i) satisfying the desired long-term inspection criteria of the sites, and ii) providing a guarantee for the UAVs to return to their base nodes for battery recharging. While ensuring that the aforementioned stochastic constraints are satisfied with high probabilities, we transform this stochastic problem into obtaining the optimal Markov chain transition matrices and the inspection policies of the UAVs. We show that the resulting problem  belongs to the category of highly non-convex signomial programming, which is in general intractable. To tackle this challenge, we first propose an optimal centralized algorithm that approximates the problem as a series of \textit{geometric programming} problems.  We demonstrate that the centralized algorithm suffers from the \textit{curse of dimensionality} upon utilizing a large number of UAVs on a map with a large number of sites. We subsequently develop a consensus-based distributed algorithm combining the dual decomposition method, the average consensus algorithm, and the gradient consensus technique.

Our solution has the following characteristics: i) it explicitly considers the limited battery capacity of the UAVs and provides a reliable stochastic surveillance, in which the UAVs can return to their bases   before their batteries run out with high probability;  ii) it gives rise to stochastic trajectories and inspection policies, which renders exact predictions by malicious entities less likely; iii) it is energy-efficient, minimizing the  long-term average energy consumption of the UAVs;  and iv) it leads to \textit{soft partitioning} of the map, different from the current map partitioning approaches (e.g.,~\cite{lin2012integrating, sastry2006instrumenting}) with disjoint islands. It forms different stochastic movement patterns for the UAVs that, in general, have partial overlaps. 
Our contributions are summarized below:

\begin{itemize}[leftmargin=*]
\vspace{-1mm}
    \item We propose a novel framework for stochastic UAV-assisted surveillance that considers \textit{random-walk movement} patterns and \textit{probabilistic inspection} patterns. Our framework explicitly considers the limited battery capacities of the UAVs and allows reliable return to their base nodes.
    \item  We formulate the  problem of joint random walk and inspection policy design (in a non-Euclidean space, i.e., a graph structure) under  probabilistic constraints as an optimization problem,  which turns out to be a highly non-convex \textit{signomial programming}  problem.
    \item  We tackle the problem by proposing a centralized algorithm based on iterative \textit{geometric programming} approximation, which exploits both the method of condensation using monomial approximations and penalty functions. We theoretically investigate the optimality of our algorithm.
    \item Given the drawbacks of the centralized algorithm, in particular  the \textit{curse of dimensionality}, we propose a \textit{distributed consensus-based algorithm}. The cornerstone of our algorithm is a combination of the dual decomposition method, the average consensus algorithm, and the gradient consensus technique. We also investigate the convergence of our distributed algorithm and demonstrate its optimality.
\end{itemize}

 It is worth noting that our proposed framework is general enough to be deployed for other applications concerned with designing energy-aware random walks, where UAV-assisted surveillance is only one use case.
More precisely, we are among the first to investigate  the design of energy-efficient random walks and inspection policies for multiple random walkers on a graph structure to satisfy the desired inspection criteria of the nodes, where each random walker has limited regenerative energy that gets renewed upon visiting its home base. The energy can also be interpreted as life-time or cost/budget in other applications such as data collection.

 The rest of this paper is structured as follows: the system model is presented in Section~\ref{sysmo}. Section~\ref{sec:markovanalysis} contains the Markov chain analysis of the stochastic movement and inspection policies and the problem formulation. The  centralized and consensus-based distributed algorithms are proposed in Section~\ref{sec:problemFormualtion}. Simulation results are presented in Section~\ref{sec:sim}. Finally, Section~\ref{sec:conc} concludes the paper and provides future directions.

\section{System Model}\label{sysmo}
\subsection{Network Model and Stochastic Inspections}
We consider designing a UAV-assisted surveillance network for a map consisting of multiple sites. Let $G=(\mathcal{V},\mathcal{E},w)$ denote the corresponding (weighted) network graph, where $\mathcal{V}=\{v_1,\cdots,v_{|\mathcal{V}|}\}$\footnote{The symbol $|.|$ denotes the cardinality of the set.} denotes the set of nodes (sites),  $\mathcal{E}$ denotes the set of edges, and $w: \mathcal{V}^2 \rightarrow \mathbb{R}^+$ is the weight function defined for nodes $v_i$ and $v_j$ as a function of the distance between the nodes and UAV parameters, which is further discussed in Appendix~\ref{Sec:costModel}.
\footnote{The weight $w(v_i, v_j)$ can also be interpreted as the cost of transition between sites $v_i$ and $v_j$.} For each node $v_i\in \mathcal{V}$, we denote by $\pi_i >0$  the required \textit{inspection criterion} on node $v_i$ and define it as the desired expected number of UAVs inspecting that site per unit time (in the steady state)\footnote{For instance, if $v_i$ desires to get inspected by $1$ UAV every $2$ time instance on average $\pi_i=0.5$. Our formulation also accommodates the scenarios where $\pi_i>1$ for more strict security purposes (e.g., when it is desired to get inspected by more than one UAV per time instant) and it does not require to normalize $\pi_i$-s. We ignore the trivial case of $\pi_i=0$ since in that case node $v_i$ needs no surveillance and can be excluded.}, and consider $\bm{\pi}= \left[\pi_1,\cdots,\pi_{|\mathcal{V}|}\right]$.

We consider a set of UAVs $\mathcal{U}=\{u_1,\cdots,u_{|\mathcal{U}|}\}$ dedicated to conduct the surveillance, and assume $|\mathcal{U}|<|\mathcal{V}|$ to avoid triviality. Each UAV is associated with a \textit{base node (home base)}, to which the UAV returns after each trip to deliver its collected data and recharge its battery. Let $v_{b^i}\in \mathcal{V}$ and $\varphi_i>0$ denote the base node and the full capacity of the battery of UAV $u_i$, respectively.  Upon arriving at a site, a UAV can perform either of the following two actions: i) inspecting the site via turning on its sensors and cameras, ii) passing the site without conducting inspection. Thus, in our model, the energy consumption of a UAV is mainly due to the following two factors: i) maintaining levitation and physical movement between the sites, and ii) turning on the mounted sensors for inspection and data collection.  For node $v_j$ and UAV $u_i$, let $I^{(i)}_j$ denote a Bernoulli random variable, where $I^{(i)}_j=1$ with probability $\iota^{(i)}_{j}\in (0,1]$ indicates turning on the sensing devices. Upon arriving at site $v_j$, UAV $u_i$ inspects the site with probability $\iota^{(i)}_{j}$ (and does not inspect with probability $1-\iota^{(i)}_{j}$), independently of all others. We define $\bm{\iota}_i=[\iota^{(i)}_{1},\cdots,\iota^{(i)}_{|\mathcal{V}|}]$ and call it as \textit{inspection policy} of UAV $u_i$.  Let $\psi^{(i)}_{j}>0$ denote the energy consumption of data collection, i.e., utilizing the camera and sensing devices, of UAV $u_i $ at node $v_j$, which can vary from one UAV to another due to heterogeneous attributes of the UAVs.
 We refer to Fig.~\ref{fig:sysmod} for illustrations. 
\subsection{Stochastic Movement of the UAVs}\label{sec:randomwalks}

 Let $X_i(t)\in \mathcal{V}$, $t=1, 2, \cdots$ be the position of UAV $u_i$ at time $t$. In our framework, $X_i(t)$ is taken to be a Markov chain with its \textit{transition matrix} $\bm{P}_i=[p^{(i)}_{kj}]_{1\leq k,j\leq |\mathcal{V}|}$, where $p^{(i)}_{kj}=\textrm{Pr}(X_i(t+1)=v_k|X_i(t)=v_j)$ is the probability of transitioning to node $v_k$ from node $v_j$.\footnote{The amount of time in transversal, which depends  on the actual distance between various nodes, is captured in different edge weights, representing the required energy spent for this traversal (see Appendix~\ref{Sec:costModel}). It can be verified that from the energy consumption perspective, skewing the notion of time in that manner has no effect on the analysis.}  The Markov chains of our interest are ergodic chains that admit unique stationary distributions, denoted by $\bm{\gamma}_{i}=[\gamma^{(i)}_1,\cdots,\gamma^{(i)}_{|\mathcal{V}|}]$ for UAV $u_i$ satisfying $\bm{\gamma}_{i} \bm{P}_i=\bm{\gamma}_{i}$. 
As will be seen later, our problem is to find the optimal \textit{transition matrices} and \textit{inspection policies} of  the UAVs to minimize their long-term average expected energy consumption while satisfying the constraints on the desired inspection criteria of the sites and the battery capacities of the UAVs. Note that the resulting trajectories of the UAVs and the inspection policies are stochastic, which, considering the existence of malicious entities, are actually desired.


\begin{figure*}[t]
\vspace{-1mm}
    \centering\includegraphics[width=0.84\linewidth, height=0.21\linewidth]{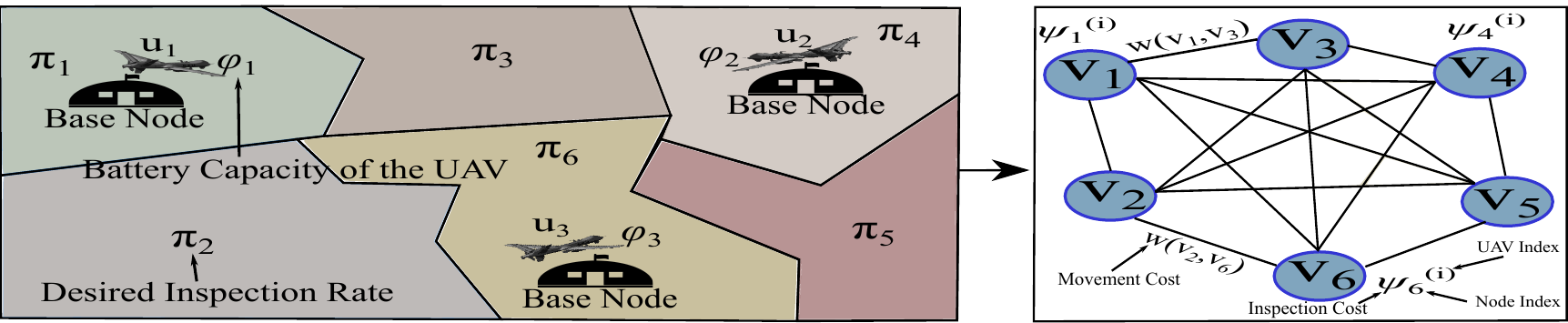}
    \caption{Geographical map of a UAV-assisted surveillance network (left) and the corresponding graph representation (right). Repetitive notations are omitted for a better readability.  By setting $w(v_i,v_j)$ of two non-adjacent nodes $v_i$ and $v_j$  large ($\to \infty$, or equivalently infinite cost) we can ensure that the resulting energy-constrained UAVs never transit between them.
   }
    \label{fig:sysmod}
     \vspace{-3mm}
\end{figure*}

    \begin{table}    
{\footnotesize
\caption{Major notations.}
    \begin{center}
\begin{tabular}{ |c |c|} 
 \hline
 Symbol& Definition\\ \hline
 $G$ & The network graph  \\ \hline
 $\mathcal{V}$ & The set of nodes/sites in the network \\ \hline
 $\mathcal{E}$ & The set of edges of the network graph \\ \hline
 $w$ & The weight function  \\ \hline
 $v_{b^i}$ & The base node of UAV $u_i\in \mathcal{U}$ \\ \hline
 $\bm{\pi}$ &The desired inspection criteria of the sites \\ \hline
 $\varphi_i$ & The full capacity of the battery of UAV $u_i$ \\ \hline
 $\mathcal{U}$ &The set of UAVs in the network\\ \hline
 $\bm{\gamma}_i$&The stationary distribution of movement of UAV $u^i\in \mathcal{U}$\\\hline
 $\bm{P}_i$&The transition matrix of UAV $u^i\in \mathcal{U}$\\ \hline
 $\psi^{(i)}_{j}$ & Energy of data collection at node $v_j\in\mathcal{V}$ for UAV $u_i\in \mathcal{U}$\\ \hline 
 $X_i(t)$& The location/position of UAV $u^i\in \mathcal{U}$ at time $t$.\\ \hline
 ${I}^{(i)}_{j}$& The Bernoulli random variable indicating the inspection\\&  of location $v_j\in\mathcal{V}$ for UAV $u^i\in \mathcal{U}$.\\ \hline
 $T^{(i)}_{+}$& A random variable corresponding to \\&the return time of the UAV $u_i$ to its base node.\\\hline
   $\bm{\iota}_{i}$ & The vector of inspection policy of UAV $u_i$ \\\hline
  $\iota^{(i)}_{j}$  & The probability of turning on the mounted sensors of $u_{i}\in \mathcal{U}$\\& for inspecting $v_j\in\mathcal{V}$; $\iota^{(i)}_{j}= p(I^{(i)}_{j}=1)$\\\hline
  $p^{(i)}_{jk}$ &  Transition probability from node $v_j\in\mathcal{V}$ \\&to node $v_k\in\mathcal{V}$ for UAV $u_{i}\in \mathcal{U}$\\\hline
  $\hat{\theta}_j$ &  Tuning  parameter  controlling  the  amount  of  violation\\& of inspection criterion for node $v_j\in\mathcal{V}$\\\hline
  $\tilde{\theta}_i$ &  Tuning  parameter  controlling  the  amount  of  violation\\& of battery capacity of UAV $u_{i}\in \mathcal{U}$\\\hline
\end{tabular}
\end{center}
}
\end{table}
\section{Probabilistic  Analysis of the Surveillance}\label{sec:markovanalysis} \label{sec:method}
\subsection{Reliable Stochastic Surveillance with Energy Constraints}\label{subsec:visitRate}
A feasible surveillance framework should fulfil the following requirements: i) the desired inspection criteria of the nodes should be satisfied; ii) the UAVs should have enough energy stored in their batteries to perform the surveillance and come back to their bases. Due to the inherent randomness in the movement of the UAVs and the inspections patterns of the sites, these constraints are of probabilistic nature. We first  state these constraints and then convert them into tractable mathematical expressions to be used in our optimization framework later on.

First, assuming the Markov chain associated with the random walks to be in its stationary regime,  satisfying the inspection criteria of the nodes can be expressed via upper bounding the probability of violation of the inspection criterion  of each node:
  \begin{equation}\label{eq:fo2}
 \textrm{Pr} \left( \sum_{i=1}^{|\mathcal{U}|} \gamma^{(i)}_j I^{(i)}_j \leq \pi_j\right) \leq \hat{\theta}_j,~\forall v_j\in \mathcal{V},
\end{equation}
 where $\hat{\theta}_j \in (0,1]$ is the tuning parameter controlling the amount of violation. Inequality \eqref{eq:fo2} implies that the rate of inspection of site $v_j$, i.e., the average number of UAVs inspecting it per unit time, is larger than $\pi_j$ with high probability (at least $1-\hat{\theta}_j$). Second, to have reliable surveillance, each UAV should be able to visit its base node for battery recharge before its battery depletion. For each UAV, we consider the random time span between its departure from and return to its base node as a \textit{surveillance cycle}. We bound the probability of exceeding the battery capacity of each UAV for each surveillance cycle:
          \begin{equation}\label{eq:for4}
          \begin{aligned}
     \textrm{Pr}\bigg(\sum_{t=1}^{T^{(i)}_{+}}w(& X_i(t), X_i(t+1))\\&+ \sum_{t=1}^{T^{(i)}_{+}}\psi^{(i)}_{X_i(t)} I^{(i)}_{X_i(t)} \geq \varphi_i \bigg)\leq \tilde{\theta}_i,~\forall u_i\in\mathcal{U},
    \end{aligned}
  \end{equation}
  where the first and the second terms inside the probability denote the total  energy consumed for the movement and for the sensing per surveillance cycle, respectively. In~\eqref{eq:for4}, $\tilde{\theta}_i \in (0,1]$ is the tuning parameter controlling the tolerable amount of violation, and $T^{(i)}_{+}$ is a random variable corresponding to the \textit{return time} of UAV $u_i$ to its base node, $T^{(i)}_{+}= \textrm{min} \{n\geq 1: X_i(n)= v_{b^i},~X_i(0)=v_{b^i}\}$.   Inequality \eqref{eq:for4} implies that the energy consumption of each UAV $u_i$ during a surveillance cycle is less than $\varphi_i$ with high probability (at least $1-\tilde{\theta}_i$).
  
 According to the strong Markov property, successive returns to a given site forms a renewal process. The energy associated with the movement of the UAV, the energy associated with using the sensing devices, or any other possible action that UAV may take during a surveillance cycle, e.g., sending and receiving data from some base stations, as the \textit{reward} during the surveillance cycle. This draws a connection between the scenario considered in this paper and the framework of the \textit{reward process} used to derive a tractable expression for~\eqref{eq:for4}.

\begin{theorem}\label{th:main}
The sufficient conditions to satisfy the probabilistic constraints given by~\eqref{eq:fo2} and \eqref{eq:for4} can be expressed based on the  stationary distribution of the Markov chains, the transition  matrices, and the inspection probabilities of the UAVs. In particular, \eqref{eq:fo2} can be transformed to:\footnote{Satisfying the inspection criteria of the nodes could also be expressed as: $  E\left[ \sum_{i=1}^{|\mathcal{U}|} \gamma^{(i)}_j I^{(i)}_j\right] \geq \pi_j,$  $\forall v_j \in \mathcal{V}$ or equivalently: $
 \sum_{i=1}^{|\mathcal{U}|}\gamma^{(i)}_j \iota^{(i)}_{j}  \geq \pi_j$, resulting in a looser bound coinciding with \eqref{ProbConst1} for $\hat{\theta}_j=1$.\label{Foot:1}}
 \begin{equation}\label{ProbConst1}
\sum_{i=1}^{|\mathcal{U}|} \gamma^{(i)}_j\iota^{(i)}_{j} \geq  \hat{\theta}_j \pi_j + \left(1-\hat{\theta}_j\right)\sum_{i=1}^{|\mathcal{U}|} \gamma^{(i)}_j, ~\forall v_j\in \mathcal{V},
 \end{equation}
and \eqref{eq:for4} can be expressed as, $\forall u_i\in\mathcal{U}$:\footnote{Constraint~\eqref{eq:for4} could also be represented as a bound on the expected value as: $E\left[\sum_{t=1}^{T^{(i)}_{+}}w(X_i(t),X_i(t+1)) +\psi^{(i)}_{X_i(t)}I^{(i)}_{X_i(t)}\right] \leq \varphi_i$, $\forall u_i\in\mathcal{U}$, the result of which is a looser bound that coincides with~\eqref{eq:battr} when $\tilde{\theta}_i=1$.}
  \begin{equation}\label{eq:battr}
              \sum_{v_j\in \mathcal{V}} \frac{\gamma^{(i)}_j}{\gamma^{(i)}_{b^i}} \sum_{v_k\in \mathcal{V}} p^{(i)}_{jk}w(v_j,v_k)+\sum_{v_j\in\mathcal{V} } \frac{\gamma^{(i)}_j}{\gamma^{(i)}_{b^i}} \psi^{(i)}_{j} \iota^{(i)}_{j} \leq \varphi_i \tilde{\theta}_i.\hspace{-1mm} 
        \end{equation}
\end{theorem}
\begin{proof}
The proof is provided in Appendix~\ref{App:ProbtoTrac}.
\end{proof}

Considering~\eqref{ProbConst1}, by setting $\iota^{(i)}_j=1$, $\forall i,j$, and taking the summation with respect to (w.r.t.) $j$ from both hand sides of the inequality, the necessary condition on the number of UAVs to satisfy~\eqref{ProbConst1} is given by $|\mathcal{U}|\geq\lceil \sum_{v_j\in\mathcal{V}} \pi_j\rceil$.
Considering UAV $u_i$, the long-term average movement energy of the UAV can be written based on its stationary distribution of the visits of the nodes and its transition  matrix as: $\displaystyle \lim_{T\longrightarrow \infty} \frac{1}{T} E\big[ \sum_{t=1}^{T}w\left(X_i\left(t\right),X_i\left(t+1\right)\right)\big]=\sum_{v_j\in \mathcal{V}}\sum_{v_k\in \mathcal{V}}  \gamma^{(i)}_j p^{(i)}_{jk}w(v_j,v_k)$. In a similar manner, the UAV's long-term average consumed energy for inspection of the nodes can be derived as:$\displaystyle \lim_{T\longrightarrow \infty} \frac{1}{T} E\bigg[ \sum_{t=1}^{T} \psi^{(i)}_{X_i(t)}I^{(i)}_{X_i(t)} \bigg]= \sum_{v_j\in\mathcal{V} } \gamma^{(i)}_j \psi^{(i)}_{j} \iota^{(i)}_{j}$. Thus, the \textit{long-term average consumed energy during the surveillance}, is given by:
 \begin{equation}\label{eq:ObjectiveFunc}
  \hspace{-3mm}
 \begin{aligned}
    \hspace{-35mm}\lim_{T\longrightarrow \infty} \frac{1}{T} E\bigg[& \sum_{t=1}^{T}\sum_{u_i\in\mathcal{U}} w\left(X_i\left(t\right),X_i\left(t+1\right)\right)+ \psi^{(i)}_{X_i(t)}I^{(i)}_{X_i(t)} \bigg]\\
           &\hspace{-17mm} =\sum_{u_i\in \mathcal{U}}\sum_{v_j\in \mathcal{V}} \sum_{v_k\in \mathcal{V}}\gamma^{(i)}_j  p^{(i)}_{jk}w(v_j,v_k)+\sum_{u_i\in \mathcal{U}} \sum_{v_j\in\mathcal{V} } \gamma^{(i)}_j \psi^{(i)}_{j} \iota^{(i)}_{j}.
     \end{aligned}\hspace{-15mm} 
 \end{equation} Due to the topological structure of the problem and heterogeneous base nodes and battery capacities of the UAVs, the optimal  transition  matrices and the inspection policies of the UAVs are different. In the following, we use the above results to formulate the problem of interest.
      \subsection{Problem Formulation}
Let us define the following sets: $\bm{P}=\{\bm{P}_1,\bm{P}_2,\cdots,\bm{P}_{|\mathcal{U}|}\},~\bm{\iota}=\{\bm{\iota}_1,\bm{\iota}_2,\cdots,\bm{\iota}_{|\mathcal{U}|}\},~ \bm{\gamma}=\{\bm{\gamma}_1,\bm{\gamma}_2,\cdots,\bm{\gamma}_{|\mathcal{U}|}\}$,
    where $\bm{P}_i$, $\bm{\iota}_i$, $\bm{\gamma}_i$ are defined as above for UAV $u_i$, $\forall u_i\in\mathcal{U}$.       The problem of interest is determining the movement and the inspection policies of the UAVs considering the aforementioned constraints. 
      This involves obtaining the above three sets. However, the elements of the two sets  $\bm{\gamma}$ and $ \bm{P}$ are not independent. More precisely, given a matrix  $ \bm{P}_i$, vector $\bm{\gamma}_i$ is uniquely defined, $\forall u_i\in\mathcal{U}$.\footnote{ $\bm{\gamma}_i$ is the left eigenvector of $\bm{P}_i$, $\forall u_i\in\mathcal{U}$.} As a result, we perform the following change of variables: 
     \begin{equation}~\label{eq:DeriveProb}
         q^{(i)}_{jk}= \gamma^{(i)}_j p^{(i)}_{jk},~ \forall u_i\in\mathcal{U},~ \forall v_j,v_k\in\mathcal{V}.
     \end{equation} 
     It is easy to verify the following two equations:
     \begin{equation}~\label{eq:DeriveGamma}
        \gamma^{(i)}_j =\sum_{v_k\in\mathcal{V}} q^{(i)}_{jk},~ \forall u_i\in \mathcal{U}, ~\forall v_j\in\mathcal{V},
     \end{equation}
     \begin{equation}
         \sum_{v_j\in\mathcal{V}} \sum_{v_k\in\mathcal{V}} q^{(i)}_{jk}=1,~\forall u_i\in \mathcal{U}.
     \end{equation}
     Note that $\sum_{v_k\in\mathcal{V}}p^{(i)}_{jk}=1$, $\forall v_j\in \mathcal{V}, u_i\in\mathcal{U}$, is implicitly satisfied via the above two equations.
     As a result, instead of finding two sets $\bm{P}$ and $\bm{\iota}$, we focus on finding set $\bm{q}$, where
     \begin{equation}
         \bm{q}=\{\bm{q}_1,\bm{q}_2,\cdots,\bm{q}_{|\mathcal{U}|}\},
     \end{equation}
     and each $\bm{q}_i=[{q}^{(i)}_{jk}]_{1\leq j,k \leq |\mathcal{V}|}$ is a matrix. For a given  $\bm{q}$,~\eqref{eq:DeriveGamma} can be used to obtain the stationary distributions of the UAVs $\bm{\gamma}$. Then, the transition  matrices of the UAVs $\bm{P}_i$-s can be obtained through~\eqref{eq:DeriveProb}.  
     Using this change of variables, with some algebraic manipulations, we formulate the \textit{energy-aware stochastic UAV-assisted surveillance} as the following optimization problem:
\begin{equation}\label{eq:originalFormualtion}
              \hspace{-2mm}
            \begin{aligned}
                  &\argmin_{\bm{q}, \bm{\iota}} \sum_{u_i\in \mathcal{U}}\sum_{v_j\in \mathcal{V}}\sum_{v_k\in \mathcal{V}} q^{(i)}_{jk} w(v_j,v_k) \\&~~~~~~~~~~~+
     \sum_{u_i\in \mathcal{U}} \sum_{v_j\in\mathcal{V} } \sum_{v_k\in\mathcal{V}} q^{(i)}_{jk} \psi^{(i)}_{j} \iota^{(i)}_{j}\\
      &\textrm{s.t.}\\
      &(\mathbf{C1})~~  \sum_{u_i\in \mathcal{U}}\sum_{v_k\in\mathcal{V}} q^{(i)}_{jk} \iota^{(i)}_{j}  \geq \hat{\theta}_j \pi_j\\&~~~~~~~~+\left(1-\hat{\theta}_j\right)\sum_{u_i\in \mathcal{U}} \sum_{v_k\in\mathcal{V}} q^{(i)}_{jk},~ \forall v_j\in \mathcal{V},\\
      &(\mathbf{C2})~~  \sum_{v_j\in\mathcal{V}} \sum_{v_k\in\mathcal{V}}q^{(i)}_{jk} w(v_j,v_k)\\&~~~~~~~~+\sum_{v_j\in\mathcal{V} } \sum_{v_k\in\mathcal{V}} q^{(i)}_{jk}\psi^{(i)}_{j} \iota^{(i)}_{j}\leq \varphi_i \tilde{\theta}_i\sum_{v_k\in\mathcal{V}} q^{(i)}_{b^i k}, \;\; \forall u_i \in \mathcal{U},\\
          &(\mathbf{C3})~~ \sum_{v_j\in \mathcal{V}} q^{(i)}_{jk} = \sum_{v_n\in\mathcal{V}} q^{(i)}_{kn}, ~\forall v_k\in
      \mathcal{V}, ~\forall u_i \in \mathcal{U},\\
      &(\mathbf{C4})~~ \sum_{v_j\in \mathcal{V}} \sum_{v_k\in\mathcal{V}} q^{(i)}_{jk} =1,~\forall u_i \in \mathcal{U},\\
     &(\mathbf{C5})~~  0< q^{(i)}_{jk},\iota^{(i)}_{j} \leq 1,~\; \forall v_j,v_k\in \mathcal{V},~ \forall u_i\in \mathcal{U}.
        \end{aligned}
        \hspace{-5mm}
        \end{equation}
 In this formulation, the objective function is the long-term average consumed energy during the surveillance given by \eqref{eq:ObjectiveFunc}. The first constraint ($\mathbf{C1}$) enforces the satisfaction of the desired inspection criteria presented in~\eqref{ProbConst1}, the second constraint ($\mathbf{C2}$) guarantees the battery consumption presented in~\eqref{eq:battr}, the third constraint ($\mathbf{C3}$) forces the stationary visiting distribution to be the left eigenvector of the transition  matrix, the forth constraint ($\mathbf{C4}$) ensures that the summation of stationary visiting distribution is equal to $1$ for each UAV, while the last constrain ($\mathbf{C5}$) is ensuring a feasible range for the variables. A noteworthy property of ($\mathbf{C2}$) is that it results in a \textit{soft  partitioning} of the map around the bases nodes of the UAVs, allowing the UAVs with small battery capacities to mostly inspect the sites located around their base nodes. Our choice of the term \textit{soft  partitioning} is due to the fact that our approach does not limit the set of nodes to do the surveillance, rather it increases the probability of surveillance for the closer nodes around the base nodes. This new perspective to the map partitioning is different from the classic map partitioning approaches in existing surveillance-related literature~\cite{lin2012integrating, sastry2006instrumenting}. Also, our method leads to a significant decrease in the number of \textit{redundant inspections}. These two facts will be further illustrated in Section~\ref{sec:simC}.
 \begin{remark}
 Note that the probabilistic problem of interest has been transformed to jointly finding a set of transition matrices of Markov chains, each of which describing the random movement of a UAV, and a set of inspection policies subject to the constraints on the inspection criteria of the nodes ($\mathbf{C1}$) and the energy consumption per surveillance cycle ($\mathbf{C2}$). To the best of our knowledge, we are among the first to propose this formulation and solve it in a systematic manner.
 \end{remark}

 \vspace{-5mm}
      \section{Optimal Random Walks and Inspection Policies}\label{sec:problemFormualtion}
\noindent  Solving the aforementioned optimization problem directly is non-trivial since the multiplication of the optimization variables exists in the formulation, e.g., between $q^{(i)}_{jk},\iota^{(i)}_{j}$ in the objective function and in the first and the second constraints. In fact, we will show that~\eqref{eq:originalFormualtion} belongs to the family of \textit{signomial programming} problems and is highly non-convex. To tackle this problem, we propose a tractable iterative approach, in which at each iteration we solve an approximation of the problem that has the format of \textit{geometric programming}~(GP). In the following, we give a brief overview of GP.
  \vspace{-5.0mm}
        \subsection{Geometric Programming}
   A basic knowledge of \textit{monomials} and \textit{posynomials}, which is given below, is a prerequisite to understand the GP. 
         \begin{definition}
         A \textbf{{monomial}} is defined as a function $f: \mathbb{R}^n_{++}\rightarrow \mathbb{R}$:\footnote{$\mathbb{R}^n_{++}$ denotes the strictly positive quadrant of $n$-dimensional Euclidean space.} $f(\bm{y})=d y_1^{\alpha_1} y_2^{\alpha_2} \cdots y_n ^{\alpha_n}$, where $d\geq 0$, $\bm{y}=[y_1,\cdots,y_n]$, and $\alpha_j\in \mathbb{R}$, $\forall j$. Further, a \textbf{posynomial} $g$ is defined as a sum of monomials: $g(\bm{y})= \sum_{m=1}^{M} d_m y_1^{\alpha^{(1)}_m} y_2^{\alpha^{(2)}_m} \cdots y_n ^{\alpha^{(n)}_m}$.
\end{definition}
A standard GP is a non-convex optimization
problem defined as minimizing a posynomial subject to posynomial 
inequality constraints and monomial equality constraints~\cite{chiang2005geometric,boyd2007tutorial}:
        \begin{equation}\label{eq:GPformat}
            \begin{aligned}
            &\min_{\bm{y}} f_0 (\bm{y})\\
            &\textrm{s.t.} ~~~ f_i(\bm{y})\leq 1, \;\; i=1,\cdots,I,\\
           &~~~~~~ h_l(\bm{y})=1, \;\; l=1,\cdots,L,
            \end{aligned}
        \end{equation}
        where  $f_i(\bm{y})=\sum_{m=1}^{M_i} d_{i,m} y_1^{\alpha^{(1)}_{i,m}} y_2^{\alpha^{(2)}_{i,m}} \cdots y_n ^{\alpha^{(n)}_{i,m}}$, $\forall i$, and $h_l(\bm{y})= d_l y_1^{\alpha^{(1)}_l} y_2^{\alpha^{(2)}_l} \cdots y_n ^{\alpha^{(n)}_l}$, $\forall l$. Since the log-sum-exp function $f(\bm{y}) = \log \sum_{j=1}^n e^{y_j}$ is convex, where $\log$ denotes the natural logarithm, with the following change of variables and constants $z_i=\log(y_i)$, $b_{i,k}=\log(d_{i,k})$, $b_l=\log (d_l)$  the GP in the convex form can be obtained as:
          \begin{equation}~\label{GPtoConvex}
            \begin{aligned}
            &\min_{\bm{z}} \;\log \sum_{m=1}^{M_0} e^{\left(\bm{\alpha}^{\top}_{0,m}\bm{z}+ b_{0,m}\right)}\\
            &\textrm{s.t.} ~~~ \log \sum_{m=1}^{M_i} e^{\left(\bm{\alpha}^{\top}_{i,m}\bm{z}+ b_{i,m}\right)}\leq 0 \;\; i=1,\cdots,I,\\
           &~~~~~~~ \bm{\alpha}_l^\top \bm{z}+b_l =0\;\; l=1,\cdots,L,
            \end{aligned}
        \end{equation}
        where $\bm{z}=[z_1,\cdots,z_n]^\top$, $\bm{\alpha}_{i,k}=\left[\alpha_{i,k}^{(1)},\alpha_{i,k}^{(2)}\cdots, \alpha_{i,k}^{(n)}\right]^\top$, $\forall i,k$, and $\bm{\alpha}_{l}=\left[\alpha_{l}^{(1)},a_{l}^{(2)}\cdots, \alpha_{l}^{(n)}\right]^\top$\hspace{-2mm}, $\forall l$.
        \vspace{-5mm}
        \subsection{Obtaining Random Walks and Inspection Policies: Centralized Approach}
         It can be verified that although the objective function and all the constraints in problem~\eqref{eq:originalFormualtion} can be expressed as monomials and posynomials w.r.t. the optimization variables, problem~\eqref{eq:originalFormualtion} does not obey the standard form of  GP in~\eqref{eq:GPformat}. In fact, none of the constraints ($\mathbf{C1}$), ($\mathbf{C2}$), ($\mathbf{C3}$), and ($\mathbf{C4}$) can be directly expressed as inequalities on posynomials or equalities on monomials, which is required in GP.\footnote{The multiplicative coefficient(s) of a posynomial/monomial should be strictly positive. For example, $-\sum_{u_i\in \mathcal{U}}\sum_{v_k\in\mathcal{V}} q^{(i)}_{jk} \iota^{(i)}_{j}$ is not a posynomial.} Thus, the problem  fits into the category of  signomial programming, for which direct derivation of a solution is intractable~\cite{chiang2005geometric}. In the following, we apply two methods, namely penalty functions and monomial approximations, on problem~\eqref{eq:originalFormualtion}, to approximate the problem as a series of GP problems. Afterward, we propose an effective algorithm to solve the problem along with its performance guarantee.
         We first use the method of penalty functions and  auxiliary variables~\cite{SignomialGlobal}. To this end, we consider each equality on a posynomial in the format of $g(\bm{x})=c$ as two inequality constraints: i) $g(\bm{x})\leq c$, and ii)~$\frac{1}{A g(\bm{x})}\leq c$, where $A$ is an auxiliary variable, which will later be forced to be $1$. Aiming to make problem~\eqref{eq:originalFormualtion} as close to a GP as possible, we perform some algebraic manipulation and rewrite it as:
         \vspace{-2mm}
         
         {\small
                    \vspace{-.5mm} \begin{equation}\label{eq:formulationAxiliary}
                     \hspace{-5mm}
            \begin{aligned}
                  &\argmin_{\bm{q}, \bm{\iota},\bm{A},\bm{B}} \sum_{u_i\in \mathcal{U}}\sum_{v_j\in \mathcal{V}}\sum_{v_k\in \mathcal{V}}  q^{(i)}_{jk} w(v_j,v_k) \\&~~~~~~~~+\sum_{u_i\in \mathcal{U}} \sum_{v_j\in\mathcal{V} } \sum_{v_k\in\mathcal{V}} q^{(i)}_{jk} \psi^{(i)}_{j} \iota^{(i)}_{j}\\&~~~~~~~~+\sum_{v_k\in \mathcal{V}}\sum_{u_i\in \mathcal{U}} w^{(A)}_{ki} A^{(i)}_{k}+ \sum_{u_i\in \mathcal{U}} w^{(B)}_{i} B_i\\
      &\textrm{s.t.}\\
      &(\mathbf{\tilde{C}1})~~  \frac{\hat{\theta}_j \pi_j+\left(1-\hat{\theta}_j\right)\displaystyle\sum_{u_i\in \mathcal{U}} \sum_{v_k\in\mathcal{V}} q^{(i)}_{jk}}{\sum_{u_i\in \mathcal{U}}\sum_{v_k\in\mathcal{V}} q^{(i)}_{jk} \iota^{(i)}_{j} }  \leq 1 ,~ \forall v_j\in \mathcal{V},\\
      &(\mathbf{\tilde{C}2})~~ \frac{\hspace{-1.2mm}\displaystyle\sum_{v_j\in\mathcal{V}} \sum_{v_k\in\mathcal{V}}q^{(i)}_{jk} w(v_j,v_k)+\hspace{-1.2mm}\sum_{v_j\in\mathcal{V} } \sum_{v_k\in\mathcal{V}} q^{(i)}_{jk}\psi^{(i)}_{j} \iota^{(i)}_{j}}{\displaystyle \varphi_i \tilde{\theta}_i\sum_{v_k\in\mathcal{V}} q^{(i)}_{b^i k}}\leq 1\\&\hspace{60mm},~ \forall u_i \in \mathcal{U},\\
       &(\mathbf{\tilde{C}3-1})~~ \frac{\sum_{v_j\in \mathcal{V}} q^{(i)}_{jk}}{ \sum_{v_n\in\mathcal{V}} q^{(i)}_{kn}}\leq 1, ~\forall v_k\in
      \mathcal{V}, \forall u_i \in \mathcal{U},\\
      &(\mathbf{\tilde{C}3-2})~~ \frac{\sum_{v_n\in\mathcal{V}} q^{(i)}_{kn}}{A^{(i)}_{k}\sum_{v_j\in \mathcal{V}} q^{(i)}_{jk}} \leq 1, ~\forall v_k\in
      \mathcal{V}, \forall u_i \in \mathcal{U},\\
      &(\mathbf{\tilde{C}4-1})~~ \sum_{v_j\in \mathcal{V}} \sum_{v_k\in\mathcal{V}} q^{(i)}_{jk} \leq 1,~\forall u_i \in \mathcal{U},\\
      &(\mathbf{\tilde{C}4-2})~~ \frac{B_{i}^{-1}}{\sum_{v_j\in \mathcal{V}} \sum_{v_k\in\mathcal{V}} q^{(i)}_{jk} }\leq 1,~\forall u_i \in \mathcal{U},\\
     &(\mathbf{\tilde{C}5})~~  0< q^{(i)}_{jk},\iota^{(i)}_{j} \leq 1,~\; \forall v_j,v_k\in \mathcal{V}, \forall u_i\in \mathcal{U},
     \\
     &(\mathbf{\tilde{C}6})~~  A^{(i)}_{k},B_i \geq 1, \; \forall v_k\in \mathcal{V}, \forall u_i\in \mathcal{U},
            \end{aligned}
            \hspace{-5mm}
        \end{equation}
        }
        \vspace{-1mm}
        
 \noindent where, in the objective function, $w^{(A)}_{ki}$ and $w^{(B)}_i$  are sufficiently large weight coefficients corresponding to the auxiliary variables $A^{(i)}_k$ and $B_i$, $\forall k,i$, respectively.  Comparing problems~\eqref{eq:originalFormualtion} and~\eqref{eq:formulationAxiliary}, it can be seen that in particular ($\mathbf{C3}$) is replaced with ($\mathbf{\tilde{C}3-1}$) and ($\mathbf{\tilde{C}3-2}$); and ($\mathbf{C4}$) is replaced with ($\mathbf{\tilde{C}4-1}$) and ($\mathbf{\tilde{C}4-2}$). The following fact about problem~\eqref{eq:formulationAxiliary} is immediate.
  
    \begin{fact}\label{fact:2}
        At the optimal point of problem~\eqref{eq:formulationAxiliary}, the auxiliary variables will be forced to take the following values: $A^{(i)}_{k}=1$, and $B_i=1$, $\forall k,i$.
        \end{fact}
    
        \begin{fact}\label{fact:1}
        Assuming $A^{(i)}_{k}=1$, and $B_i=1$, $\forall k,i$, the solution of~\eqref{eq:formulationAxiliary} coincides with the solution of \eqref{eq:originalFormualtion}.
        \end{fact}

       
Problem~\eqref{eq:formulationAxiliary} still does not obey the standard GP form since the left hand sides (l.h.s) of ($\mathbf{\tilde{C}1}$), ($\mathbf{\tilde{C}2}$), ($\mathbf{\tilde{C}3-1}$), ($\mathbf{\tilde{C}3-2}$), ($\mathbf{\tilde{C}4-2}$) are ratios of two posynomials, instead of posynomials. Our next goal is to find favorable tight approximations for these constrains. To this end, we utilize the \textit{arithmetic-geometric mean inequality} given in the following lemma.
 \vspace{-1mm}
         \begin{lemma}[\textbf{Arithmetic-geometric mean inequality}~\cite{duffin1972reversed,chiang2005geometric}]\label{Lemma:ArethmaticGeometric}
         Consider a posynomial function $g(\bm{y})=\sum_{k=1}^{K} u_k(\bm{y})$, where $u_k(\bm{y})$ is a monomial, $\forall k$. The following inequality holds:
         \begin{equation}\label{eq:approxPosMon}
             g(\bm{y})\geq \hat{g}(\bm{y})\triangleq \prod_{k=1}^{K}\left( \frac{u_k(\bm{y})}{\alpha_k(\bm{z})}\right)^{\alpha_k(\bm{z})},
         \end{equation}
         where $\alpha_k(\bm{z})=u_k(\bm{z})/g(\bm{z})$, $\forall k$, and $\bm{z}>0$ is a fixed point.
         \end{lemma}
  \begin{algorithm}[t]
 	\caption{Centralized random walk and inspection policy design of stochastic surveillance}\label{alg:cent}
 	\SetKwFunction{Union}{Union}\SetKwFunction{FindCompress}{FindCompress}
 	\SetKwInOut{Input}{input}\SetKwInOut{Output}{output}
 	 	{\footnotesize
 	\Input{Convergence criterion.}
 	 Initialize the iteration count $m=0$.\\
 	 Choose an initial feasible point $\bm{x}^{[0]}=[\bm{q}^{[0]},\bm{\iota}^{[0]}]$.\\
 	 Obtain the monomial approximations given in~\eqref{eq:app1}-\eqref{eq:app3}.\label{midAlg1}\\
 	 Replace those approximations in~\eqref{eq:formulationAxiliaryRelaxed}.\\
 	 Using the logarithmic change of variables and taking the log from constraints, convert the GP programming in~\eqref{eq:formulationAxiliaryRelaxed} to a convex optimization problem in the form of~\eqref{GPtoConvex}.\\
 	 $m=m+1$\\
 	 Solve the resulting convex optimization problem using an arbitrary tool (e.g., CVX~\cite{grant2014cvx}) to obtain the solution $\bm{x}^{[m]}$.\label{Alg:Gpconvexste}\\
 	 \If{the convergence criterion between two consecutive solutions $\bm{x}^{[m-1]}$ and $\bm{x}^{[m]}$ is not met}{
 	\textrm{Go to line~\ref{midAlg1} and repeat the procedure using $\bm{x}^{[m]}$.}\\\Else{Choose the obtained point as the final solution $\bm{x}^{*}=\bm{x}^{[m]}=[\bm{q}^{*},\bm{\iota}^{*}]$.\label{Alg:point2}\\
 	 Replace the values of $\bm{q}^{*}$ in~\eqref{eq:DeriveGamma} to obtain the optimal stationary distribution of the UAVs $\bm{\gamma}^*$.\\
 	  Replace the values of $\bm{\gamma}^*$ in~\eqref{eq:DeriveProb} to obtain the optimal transition  matrices of the UAVs $\bm{P}^*$.\\}}
 	  }
  \end{algorithm}
         We solve problem~\eqref{eq:formulationAxiliary} via an iterative approach, in which the solution at the $m$-th iteration $\bm{x}^{[m]}=[\bm{q}^{[m]}, \bm{\iota}^{[m]}]$ is obtained based on the solution of the previous iteration $\bm{x}^{[m-1]}$. We define ${h}_j(\bm{x}),{r}^{(i)}(\bm{x}),  {z}^{(i)}_{k}(\bm{x}), f^{(i)}_k(\bm{x}),{b}^{(i)}(\bm{x})$ as the denominators of ($\mathbf{\tilde{C}1}$), ($\mathbf{\tilde{C}2}$),
          ($\mathbf{\tilde{C}3-1}$),
          ($\mathbf{\tilde{C}3-2}$),
          ($\mathbf{\tilde{C}4-2}$), respectively. At each iteration, $m$, we approximate them using the obtained solution in the previous iteration, $\bm{x}^{[m-1]}$, via \textit{arithmetic-geometric mean inequality}, the result of which is given in \eqref{eq:app1}-\eqref{eq:app3}.  It is easy to verify that $\hat{h}_j(\bm{x}),\hat{r}^{(i)}(\bm{x}), \hat{z}^{(i)}_{k}(\bm{x}),\hat{f}^{(i)}_k(\bm{x}),\hat{b}^{(i)}(\bm{x})$ given in \eqref{eq:app1}-\eqref{eq:app3} are in fact the best local monomial approximations to their corresponding posynomials
near fixed point $\bm{x}^{[m-1]}$ in the sense of the first-order Taylor approximation. Note that the ratio between a posynomial (e.g., the numerators of the aforementioned constraints) and a monomial (e.g., the corresponding monomial approximations of their denominators) is a posynomial. As a result, we can approximate these constraints as inequalities on posynomials at each iteration. Finally, we write the problem as~\eqref{eq:formulationAxiliaryRelaxed} and present the pseudo code of our proposed algorithm in Algorithm~\ref{alg:cent}. The objective function and all the constraints in the following formulation obey the standard GP format. Also, the optimality of our algorithm is investigated in Proposition~\ref{propKKT2}.
           \begin{table*}[tbp]
\begin{minipage}{0.99\textwidth}
{\footnotesize
        \begin{equation}\label{eq:app1}
       h_j(\bm{x})= \sum_{u_i\in \mathcal{U}}\sum_{v_k\in\mathcal{V}} q^{(i)}_{jk} \iota^{(i)}_{j}~\Rightarrow~ h_j(\bm{x}) \geq  \hat{h}_j(\bm{x})
           \triangleq \prod_{u_i\in \mathcal{U}} \prod_{v_k\in\mathcal{V}}  \left(\frac{q^{(i)}_{jk} \iota^{(i)}_{j}h_j(\bm{x}^{[m-1]})}{\{q^{(i)}_{jk}\}^{[m-1]}\{\iota^{(i)}_{j}\}^{[m-1]}}) \right)^\frac{\{q^{(i)}_{jk}\}^{[m-1]}\{\iota^{(i)}_{j}\}^{[m-1]}}{h_j(\bm{x}^{[m-1]})} 
        \end{equation}
             \begin{equation}\label{eq:app22}
       r^{(i)}(\bm{x})=\sum_{v_k\in\mathcal{V}} q^{(i)}_{b^i k} ~\Rightarrow~  r^{(i)}(\bm{x}) \geq \hat{r}^{(i)}(\bm{x}) \triangleq  \prod_{v_k\in\mathcal{V}} \left(\frac{q^{(i)}_{b^i k} r^{(i)}(\bm{x}^{[m-1]})}{\{q^{(i)}_{b^i k}\}^{[m-1]}}\right)^{\frac{\{q^{(i)}_{b^i k}\}^{[m-1]}}{r^{(i)}(\bm{x}^{[m-1]})}}
        \end{equation}
        \vspace{-1mm}
            \begin{equation}\label{eq:app2}
      z^{(i)}_{k}(\bm{x})=\sum_{v_n\in \mathcal{V}} q^{(i)}_{kn} ~\Rightarrow~ z^{(i)}_{k}(\bm{x}) \geq \hat{z}^{(i)}_{k}(\bm{x}) 
           \triangleq \prod_{v_n\in \mathcal{V}} \left(\frac{q^{(i)}_{kn} z^{(i)}_{k}(\bm{x}^{[m-1]})}{\{q^{(i)}_{kn}\}^{[m-1]}}\right)^{\frac{\{q^{(i)}_{kn}\}^{[m-1]}}{z^{(i)}_{k}(\bm{x}^{[m-1]})}} 
        \end{equation}
         \vspace{-2mm}
            \begin{equation}\label{eq:app222}
       f^{(i)}_k(\bm{x})=\sum_{v_j\in \mathcal{V}} q^{(i)}_{jk} ~\Rightarrow~  f^{(i)}_k(\bm{x}) \geq \hat{f}^{(i)}_{k}(\bm{x}) 
           \triangleq \prod_{v_j\in \mathcal{V}} \left(\frac{q^{(i)}_{jk} f^{(i)}_k(\bm{x}^{[m-1]})}{\{q^{(i)}_{jk}\}^{[m-1]}}\right)^{\frac{\{q^{(i)}_{jk}\}^{[m-1]}}{f^{(i)}_k(\bm{x}^{[m-1]})}} 
        \end{equation}
         \vspace{-3mm}
            \begin{equation}\label{eq:app3}
        b^{(i)}(\bm{x})=\sum_{v_j\in \mathcal{V}} \sum_{v_k\in\mathcal{V}} q^{(i)}_{jk}~\Rightarrow~  b^{(i)}(\bm{x}) \geq \hat{b}^{(i)}(\bm{x}) 
           \triangleq \prod_{v_j\in \mathcal{V}} \prod_{v_k\in \mathcal{V}} \left(\frac{b^{(i)}(\bm{x}^{[m-1]})q^{(i)}_{jk}}{\{q^{(i)}_{jk}\}^{[m-1]}} \right)^\frac{\{q^{(i)}_{jk}\}^{[m-1]}}{b^{(i)}(\bm{x}^{[m-1]})} 
        \end{equation}
        }
        \hrulefill
\end{minipage}
\end{table*} 
\vspace{-5mm}

{\small
     \begin{equation}\label{eq:formulationAxiliaryRelaxed}
               \hspace{-15mm}
            \begin{aligned}
                  &\argmin_{\bm{q}, \bm{\iota},\bm{A},\bm{B}} \sum_{u_i\in \mathcal{U}}\sum_{v_j\in \mathcal{V}}\sum_{v_k\in \mathcal{V}}  q^{(i)}_{jk} w(v_j,v_k) \\&~~~~~~~~+\sum_{u_i\in \mathcal{U}} \sum_{v_j\in\mathcal{V} } \sum_{v_k\in\mathcal{V}} q^{(i)}_{jk} \psi^{(i)}_{j} \iota^{(i)}_{j}\\&~~~~~~~~+\sum_{v_k\in \mathcal{V}}\sum_{u_i\in \mathcal{U}} w^{(A)}_{ki} A^{(i)}_{k}+ \sum_{u_i\in \mathcal{U}} w^{(B)}_{i} B_i\\
      &\textrm{s.t.}\\
      &(\mathbf{\hat{C}1})~~  \frac{\hat{\theta}_j \pi_j+\left(1-\hat{\theta}_j\right)\displaystyle\sum_{u_i\in \mathcal{U}} \sum_{v_k\in\mathcal{V}} q^{(i)}_{jk}}{\hat{h}_j(\bm{x})}  \leq 1 ,~ \forall v_j\in \mathcal{V},\\
      &(\mathbf{\hat{C}2})~~ \frac{\hspace{-1.2mm}\displaystyle\sum_{v_j\in\mathcal{V}}\hspace{-.2mm} \sum_{v_k\in\mathcal{V}}\hspace{-.6mm}q^{(i)}_{jk} w(v_j,v_k)\hspace{-.5mm}+\hspace{-1.2mm}\sum_{v_j\in\mathcal{V} }\hspace{-.2mm} \sum_{v_k\in\mathcal{V}}\hspace{-.6mm} q^{(i)}_{jk}\psi^{(i)}_{j} \iota^{(i)}_{j}}{ \varphi_i \tilde{\theta}_i\hat{r}^{(i)}(\bm{x}) }\leq 1\\&\hspace{60mm},\;\; \forall u_i \in \mathcal{U},\\
       &(\mathbf{\hat{C}3-1})~~ \frac{\sum_{v_j\in \mathcal{V}} q^{(i)}_{jk}}{ \hat{z}^{(i)}_{k}(\bm{x}) }\leq 1, ~\forall v_k\in
      \mathcal{V}, \forall u_i \in \mathcal{U},\\
      &(\mathbf{\hat{C}3-2})~~ \frac{\sum_{v_n\in\mathcal{V}} q^{(i)}_{kn}}{A^{(i)}_{k}\hat{f}^{(i)}_{k}(\bm{x}) } \leq 1, ~\forall v_k\in
      \mathcal{V}, \forall u_i \in \mathcal{U},\\
      &(\mathbf{\hat{C}4-1})~~ \sum_{v_j\in \mathcal{V}} \sum_{v_k\in\mathcal{V}} q^{(i)}_{jk} \leq 1,~\forall u_i \in \mathcal{U},\\
      &(\mathbf{\hat{C}4-2})~~ \frac{B_{i}^{-1}}{ \hat{b}^{(i)}(\bm{x}) }\leq 1,~\forall u_i \in \mathcal{U},\\
     &(\mathbf{\hat{C}5})~~  0< q^{(i)}_{jk},\iota^{(i)}_{j} \leq 1,~\; \forall v_j,v_k\in \mathcal{V}, \forall u_i\in \mathcal{U},
     \\
     &(\mathbf{\hat{C}6})~~  A^{(i)}_{k},B_i \geq 1, \; \forall v_k\in \mathcal{V}, \forall u_i\in \mathcal{U}.
            \end{aligned}
            \hspace{-12mm}
        \end{equation}
        }
        
          \begin{table*}[t]
\begin{minipage}{0.99\textwidth}
          {\footnotesize     \begin{equation}\label{eq:longlag1}
          \hspace{-6mm}
          \begin{aligned}
            &L(\tilde{\bm{x}},\bm{\lambda},\bm{\gamma},\bm{\phi},\bm{\rho},\bm{\xi},\bm{\delta},\bm{\upsilon},\bm{\chi},\bm{\beta},\bm{\varepsilon}) =\sum_{u_i\in \mathcal{U}}\underbrace{\sum_{v_j\in \mathcal{V}}\sum_{v_k\in \mathcal{V}}  \exp\left(\tilde{q}^{(i)}_{jk}\right) w(v_j,v_k) }_{T^{(i)}_1}+
                  \sum_{u_i\in \mathcal{U}}
                  \underbrace{\sum_{v_j\in\mathcal{V} } \sum_{v_k\in\mathcal{V}} \exp\left(\tilde{q}^{(i)}_{jk}\right) \psi^{(i)}_{j} \exp\left(\tilde{\iota}^{(i)}_j\right)}_{T^{(i)}_2}
                  \\[-0.5em]
      &+\sum_{u_i\in \mathcal{U}} \underbrace{\sum_{v_k\in \mathcal{V}}w^{(A)}_{ki} \exp\left(\tilde{A}^{(i)}_{k}\right)}_{T^{(i)}_3}+ \sum_{u_i\in \mathcal{U}}\underbrace{ w^{(B)}_{i} \exp\left(\tilde{B_i}\right)}_{T^{(i)}_4}+\sum_{ v_j\in \mathcal{V}} \lambda_j \log\left({ \pi_j}\right) \\[-1.5em]
      &-\sum_{ v_j\in \mathcal{V}} \lambda_j\sum_{u_i\in \mathcal{U}} \sum_{v_k\in\mathcal{V}} \log\left( \left(\frac{\exp(\tilde{q}^{(i)}_{jk}) \exp(\tilde{\iota}^{(i)}_j) h_j(\tilde{\bm{x}}^{[m-1]})}{\exp\left(\{\tilde{q}^{(i)}_{jk}\}^{[m-1]}\right)\exp\left(\{\tilde{\iota}^{(i)}_j\}^{[m-1]}\right)} \right)^\frac{\exp\left(\{\tilde{q}^{(i)}_{jk}\}^{[m-1]}\right)\exp\left(\{\tilde{\iota}^{(i)}_j\}^{[m-1]}\right)}{h_j(\tilde{\bm{x}}^{[m-1]})} \right)\\[-0.2em]
      &+\sum_{u_i \in \mathcal{U}}\underbrace{ \zeta_i\log\left( \frac{\displaystyle\sum_{v_j\in\mathcal{V}} \sum_{v_k\in\mathcal{V}} \exp\left(\tilde{q}^{(i)}_{jk}\right) w(v_j,v_k)+\sum_{v_j\in\mathcal{V} } \sum_{v_k\in\mathcal{V}} \exp\left(\tilde{q}^{(i)}_{jk}\right)\psi^{(i)}_{j} \exp\left(\tilde{\iota}^{(i)}_j\right)}{\hat{r}^{(i)}(\tilde{\bm{x}}) }\right)}_{T^{(i)}_5}\\[-0.7em]
       &+\sum_{u_i \in \mathcal{U}}\underbrace{\sum_{v_k\in
      \mathcal{V}}  \phi_{k,i} \log\left(\frac{\sum_{v_j\in \mathcal{V}} \exp\left(\tilde{q}^{(i)}_{jk}\right)}{ \hat{z}^{(i)}_{k}(\tilde{\bm{x}}) }\right)}_{T^{(i)}_6}+ \sum_{u_i \in \mathcal{U}}\underbrace{ \sum_{v_k\in
      \mathcal{V}}\rho_{k,i}\log\left( \frac{ \sum_{v_n\in\mathcal{V}} \exp\left(\tilde{q}^{(i)}_{kn}\right)}{\exp\left(\tilde{A}^{(i)}_{k}\right)\hat{f}^{(i)}_{k}(\tilde{\bm{x}}) }\right)}_{T^{(i)}_7}+\sum_{u_i \in \mathcal{U}}\underbrace{  \xi_{i}\log\left( \sum_{v_j\in \mathcal{V}} \sum_{v_k\in\mathcal{V}} \exp\left(\tilde{q}^{(i)}_{jk}\right)\right)}_{T^{(i)}_8}\\[-0.1em]
     &+\sum_{u_i \in \mathcal{U}}\underbrace{  \delta_i\log\left( \frac{1}{\exp\left(\tilde{B}_{i}\right) \hat{b}^{(i)}(\tilde{\bm{x}}) }\right)}_{T^{(i)}_9}+ \sum_{u_i\in \mathcal{U}}\underbrace{\sum_{v_j \in \mathcal{V}}\sum_{v_k\in \mathcal{V}}  \upsilon_{j,k,i} \tilde{q}^{(i)}_{jk}}_{T^{(i)}_{10}} +  \sum_{u_i\in \mathcal{U}}\underbrace{\sum_{v_k\in \mathcal{V}}\chi_{k,i} \tilde{\iota}^{(i)}_k}_{T^{(i)}_{11}}
      -\sum_{u_i\in \mathcal{U}} \underbrace{\sum_{v_k\in \mathcal{V}} \beta_{k,i} \tilde{A}^{(i)}_{k}}_{T^{(i)}_{12}} - \sum_{u_i\in \mathcal{U}}\underbrace{ \varepsilon_{i} \tilde{B}_i}_{T^{(i)}_{13}}
                    \end{aligned}\hspace{-10mm}
          \end{equation}
          }
          \hrulefill
           {\footnotesize    
        \vspace{-.1mm}\begin{equation}\label{eq:longlag2}
          \hspace{-20mm}
          \begin{aligned}
            &L_i(\tilde{\bm{x}},\bm{\lambda},\bm{\gamma},\bm{\phi},\bm{\rho},\bm{\xi},\bm{\delta},\bm{\upsilon},\bm{\chi},\bm{\beta},\bm{\varepsilon}) =
                  \sum_{j=1}^{13} T^{(i)}_j+\sum_{ v_j\in \mathcal{V}} \frac{\lambda^{(i)}_{j}}{|\mathcal{U}|} \log\left({\pi_j}\right) \\[-0.1em]
      &+ \sum_{ v_j\in \mathcal{V}} \sum_{v_k\in\mathcal{V}}\lambda^{(i)}_{j} \frac{\exp\left(\{\tilde{q}^{(i)}_{jk}\}^{[m-1]}+\{\tilde{\iota}^{(i)}_j\}^{[m-1]}\right)}{h_j(\tilde{\bm{x}}^{[m-1]})} \left[ \{\tilde{q}^{(i)}_{jk}\}^{[m-1]}+\{\tilde{\iota}^{(i)}_{j}\}^{[m-1]}-\tilde{q}^{(i)}_{jk} -\tilde{\iota}^{(i)}_j -\log\left(h_j(\tilde{\bm{x}}^{[m-1]}) \right)\right]
                    \end{aligned}
                    \hspace{-20mm}
          \end{equation}
          }
          \hrulefill
\end{minipage}
\vspace{1mm}
\end{table*}

 \begin{algorithm}[t]
 {\footnotesize
 	\caption{Consensus-based distributed random walk and inspection policy design of stochastic surveillance}\label{alg:fulldist}
 	\SetKwFunction{Union}{Union}\SetKwFunction{FindCompress}{FindCompress}
 	\SetKwInOut{Input}{input}\SetKwInOut{Output}{output}
 	\Input{Convergence criterion.}
    Initialize the iteration count m=0.\\
    Initialize a feasible solution for the problem $\{\tilde{q}^{(i)}_{jk}\}^{[0]}$, $\{\tilde{\iota}^{(i)}_{j}\}^{[0]}$, $\{\tilde{A}^{(i)}_{k}\}^{[0]}$, $\{\tilde{B}_i\}^{[0]}$, $\forall i,j,k$.\\
     \While{The convergence criterion between two consecutive solutions of the problem $\tilde{\bm{x}}^{[m-1]}$ and $\tilde{\bm{x}}^{[m]}$ is not achieved \textbf{OR} $m=0$}{
        At each processor $i$, obtain the value of \\
    \hspace{-1mm}$h^{(i)}_{j}(\tilde{\bm{x}}^{[m]})= \displaystyle\sum_{v_k\in\mathcal{V}} \exp\left(\{\tilde{q}^{(i)}_{jk}\}^{[m]}+ \{\tilde{\iota}^{(i)}_j\}^{[m]}\right)$,~$\forall v_j\in \mathcal{V}$.
    \\
    Obtain $h_j(\tilde{\bm{x}}^{[m]})= \sum_{u_i\in \mathcal{U}}h_{j,i}(\tilde{\bm{x}}^{[m]})$, $\forall j$, at all the processors using a distributed average consensus method.\\
     Set $m=m+1$, $t=0$, and initialize the Lagrangian multiplier vectors/matrices with their elements being set to $0$.\\
     Set $\{\tilde{q}^{(i)}_{jk}\}^{[m]}_{[0]}=\{\tilde{q}^{(i)}_{jk}\}^{[m-1]}$, $\{\iota^{(i)}_{j}\}^{[m]}_{[0]}=\{\iota^{(i)}_{j}\}^{[m-1]}$, $\{\tilde{A}^{(i)}_{k}\}^{[m]}_{[0]}=\{\tilde{A}^{(i)}_{k}\}^{[m-1]}$, $\{\tilde{B}_i\}^{[m]}_{[0]}=\{\tilde{B}_i\}^{[m-1]}$, $\forall i,j,k$.\\
    \While{The convergence of the Lagrangian multipliers between two consecutive iterations is not achieved \textbf{OR} $t=0$}{ 
     \nonl $\backslash \backslash$ \textbf{Solving the inner problem:}\\
    At each processor $i$, set the current values of the  Lagrangian multipliers in the respective term in~\eqref{eq:longlag2}.\\
    At each processor $i$, derive $\{\tilde{q}^{(i)}_{jk}\}^{[m]}_{[t+1]}$, $\{\tilde{\iota}^{(i)}_{j}\}^{[m]}_{[t+1]}$, $\{\tilde{A}^{(i)}_{k}\}^{[m]}_{[t+1]}$, $\{\tilde{B}_i\}^{[m]}_{[t+1]}$, $\forall j,k$ by applying the gradient descent method on the respective Lagrangian term~\eqref{eq:longlag2}.\\
    \nonl $\backslash \backslash$ \textbf{Solving the outer problem:}\\
    Using the above obtained values,
    at each processor $i$, obtain ${\lambda'^{(i)}_j}^{[t+1]},{\zeta_i}^{[t+1]},{\phi_i}^{[t+1]},{\rho_i}^{[t+1]},{\xi_i}^{[t+1]},{\delta_i}^{[t+1]},{\upsilon_i}^{[t+1]}$  $,{\chi_i}^{[t+1]},{\beta_i}^{[t+1]},{\varepsilon_i}^{[t+1]}$ locally via the gradient ascent method applied on the respective  term of~\eqref{eq:longlag2}.\\At each processor $i$, update ${\lambda^{(i)}_j}^{[t+1]}$ using~\eqref{eq:full_dist_update}.\\
    $t=t+1$
    }
  \hspace{-2mm}The current solution is given by: $\{\tilde{A}^{(i)}_{k}\}^{[m]}=\{\tilde{A}^{(i)}_{k}\}^{[m]}_{[t]}$,~$\forall i,k$,\\ \hspace{-3mm}$\{\tilde{B}_i\}^{[m]}=\{\tilde{B}_i\}^{[m]}_{[t]}$, $\forall i$, $\tilde{\bm{q}}^{[m]}=[{\{\tilde{q}^{(i)}_{jk}\}}^{[m]}_{[t]}]_{1\leq i\leq |\mathcal{U}|, 1\leq j,k\leq |\mathcal{V}|}$,\\ \hspace{-3mm} $\tilde{\bm{\iota}}^{[m]}=[{\{\tilde{\iota}^{(i)}_{j}\}}^{[m]}_{[t]}]_{1\leq i\leq |\mathcal{U}|, 1\leq j\leq |\mathcal{V}|}$,$\tilde{\bm{x}}^{[m]}=[\tilde{\bm{q}}^{[m]},\tilde{\bm{\iota}}^{[m]}]$.
    }
    Choose the obtained point as the final solution $\tilde{\bm{x}}^*=\tilde{\bm{x}}^{[m]}$.\\
 	 Obtain the values of $\bm{q}^*$ using ${\tilde{\bm{q}}}^*$ and replace them in~\eqref{eq:DeriveGamma} to obtain the optimal stationary distribution of the UAVs $\bm{\gamma}^*$.\\
 	  Obtain the values of ${\bm{\iota}}^*$ using $\tilde{\bm{\iota}}^*$ and replace them in~\eqref{eq:DeriveProb} to obtain the optimal transition  matrices of the UAVs $\bm{P}^*$.
    }
 \end{algorithm}


        \begin{proposition}\label{propKKT2}
         Algorithm~\ref{alg:cent} generates a sequence of improved
feasible solutions that converge to a point $\bm{x}^*$
satisfying the KKT conditions of the original problem formulation~\eqref{eq:originalFormualtion}.
        \end{proposition}
         \vspace{-3mm}
      \begin{proof}
      The proof is provided in Appendix~\ref{app:1}.
      \end{proof}
      \begin{remark}
      The time taken for the Markov chains associated with the movements of the UAVs to approach their stationary distributions can be studied by analyzing their \textit{mixing times} that can be characterized based on the obtained transition matrices of the UAVs~\cite{levin2017markov,aldous}.
      \end{remark}
         \vspace{-1mm}
        \subsection{Obtaining Random Walks and Inspection Policies: Distributed Approach}\label{sebsec:Dist}
        So far, our proposed method is a centralized approach, which asks for a powerful centralized processor to obtain the solution. This raises two concerns: i) the processor should have a global knowledge about all the parameters of all the UAVs, i.e., at each iteration, it should have the knowledge of the entire sets $\bm{q}=[\bm{q}_1,\bm{q}_2,\cdots,\bm{q}_{|\mathcal{U}|}]$ and $\bm{\iota}=[\bm{\iota}_1,\bm{\iota}_2,\cdots,\bm{\iota}_{|\mathcal{U}|}]$; and the inspection and transition costs of all the UAVs. This is due to the fact that solving the convex programming in line~\ref{Alg:Gpconvexste} of Algorithm~\ref{alg:cent} is carried out using the gradient descent method~\cite{boyd}, in which the iterative update of the Lagrangian multipliers requires global knowledge of the current values of all the Lagrangian multipliers and all the UAVs' parameters.  Obtaining this knowledge might be cumbersome/infeasible in some scenarios. ii) The size of the problem is $|\mathcal{U}|\times \left( |\mathcal{V}|^2+|\mathcal{V}|\right)$, which, for a given map, escalates quickly as the number of UAVs increases.\footnote{The auxiliary variables are ignored. For each UAV, we need to obtain a $|\mathcal{V}| \times |\mathcal{V}|$ transition matrix and an inspection probability vector of size $|\mathcal{V}|$.} Based on the above two considerations, we aim to develop a distributed algorithm that firstly eliminates the requirement on global knowledge and secondly is scalable w.r.t. the number of UAVs. To achieve this, we aim to break down the problem into $|\mathcal{U}|$ individual sub-problems, each of which can be solved using a single \textit{processor}. A processor can refer to the UAV's computing devices, the computing facilities of the UAV's base node, or any third party computing facility.
        
        Nevertheless, in our case, breaking down the problem is not trivial due to the coupled structure of the Lagrangian function.
        Given problem~\eqref{eq:formulationAxiliaryRelaxed}, consider the following change of variables: $\tilde{q}^{(i)}_{jk} = \log \left({q}^{(i)}_{jk} \right)$, $\tilde{\iota}^{(i)}_j= \log \left({\iota}_i(v_j) \right)$, $\tilde{A}^{(i)}_{k}= \log \left({A}_{ki} \right)$ $\tilde{B}_i= \log \left( {B}_i\right)$, and $\tilde{\bm{x}}= [\tilde{\bm{q}}, \tilde{\bm{\iota}}]$, where $\tilde{\bm{q}}, \tilde{\bm{\iota}}$ are defined similar to ${\bm{q}}, {\bm{\iota}}$ considering the new variables. Writing problem~\eqref{eq:formulationAxiliaryRelaxed} w.r.t. these new variables and taking the $\log$ from all the inequality constraints will result in a convex programming problem. We omit the resulting problem in the interest of space; however, we derive the \textit{Lagrangian function} of the problem in~\eqref{eq:longlag1}, which is of particular interest, where $\bm{\lambda},\bm{\zeta},\bm{\phi},\bm{\rho},\bm{\xi},\bm{\delta},\bm{\upsilon},\bm{\chi},\bm{\beta}$, and $\bm{\varepsilon}$ are the vector/matrix of Lagrangian multipliers. Note that we assumed $\hat{\theta}_j\rightarrow 1$, $\forall i$, for convenience (see Footnote~\ref{Foot:1}). Also, $h_j(\bm{x})$ is expanded using~\eqref{eq:app1} since except $h_j(\bm{x})$ all the monomial approximations are functions of the UAV index $i$ and can be locally computed without the knowledge of other UAV parameters. This expansion results in the last term in the second line and the term in the third line of~\eqref{eq:longlag1}.  The corresponding \textit{dual function} of the problem can be written as:
  \begin{align}\label{eq:dualfunc}
  \begin{aligned}
  D(\bm{\lambda},\bm{\zeta},\bm{\phi},\bm{\rho},&\bm{\xi},\bm{\delta},\bm{\upsilon},\bm{\chi},\bm{\beta},\bm{\varepsilon})\\
  &= \min_{\tilde{\bm{x}}}{L(\tilde{\bm{x}},\bm{\lambda},\bm{\zeta},\bm{\phi},\bm{\rho},\bm{\xi},\bm{\delta},\bm{\upsilon},\bm{\chi},\bm{\beta},\bm{\varepsilon})}.
  \end{aligned}
  \end{align}
 Consequently, the \textit{dual problem} is given by:
 \begin{align}\label{ref:dualProb}
  \max_{\bm{\lambda},\bm{\zeta},\bm{\phi},\bm{\rho},\bm{\xi},\bm{\delta},\bm{\upsilon},\bm{\chi},\bm{\beta},\bm{\varepsilon}} D(\bm{\lambda},\bm{\zeta},\bm{\phi},\bm{\rho},\bm{\xi},\bm{\delta},\bm{\upsilon},\bm{\chi},\bm{\beta},\bm{\varepsilon}),
 \end{align}
 where all the elements of each Lagrangian vector/matrix is assumed to be in ${\scriptsize{ \mathbb{R}^+}}$.
 Since the problem in hand is a convex optimization problem in standard form, the \textit{duality gap} is zero. As a result, the solution of problem~\eqref{eq:formulationAxiliaryRelaxed} coincides with the solution of problem~\eqref{ref:dualProb}. Problem~\eqref{ref:dualProb} can be considered as a max-min optimization problem.  Throughout, we recall the inner minimization problem, i.e., deriving the dual function using~\eqref{eq:dualfunc}, as the \textit{inner-problem} and the outer maximization problem, i.e., deriving the Lagrangian multipliers, as the \textit{outer-problem}. Based on the convexity of the original problem, which directly results in the concavity of the dual problem~\cite{boyd}, this max-min problem can be solved iteratively by assuming a set of Lagrangian multipliers for the inner problem to obtain the corresponding solution, i.e., $\tilde{\bm{x}}$, then replacing $\tilde{\bm{x}}$ in the Lagrangian function to obtain the dual-function and solve the outer problem to find the new Lagrangian multipliers. This process can be repeated until the convergence occurs. Nonetheless, since the Lagrangian function is not a separable function w.r.t. the UAV indices, this process cannot be carried out in a distributed fashion in a straightforward manner. In the following, considering~\eqref{eq:longlag1} we present two observations, which are of particular interest and are the cornerstones of our proposed consensus-based distributed algorithm.
 

\vspace{-.2mm}
 \textbf{Observation 1:} Assuming known values for the Lagrangian multipliers, except for the term on the third line, all of the terms are either known  or can be broken down w.r.t. the index of UAVs. However, in the term in the third line, $h_j(\tilde{\bm{x}}^{[m-1]})$ appears inside the argument of the $\log$ function, where $h_j(\tilde{\bm{x}}^{[m-1]})= \sum_{u_i\in \mathcal{U}}\sum_{v_k\in\mathcal{V}} \exp{\left(\{\tilde{q}^{(i)}_{jk}\}^{[m-1]}+\{\tilde{\iota}^{(i)}_j\}^{[m-1]}\right)}$, making the term inseparable w.r.t. the index $i$. 

\textbf{Observation 2:} Assuming a solution $\tilde{\bm{x}}=[\tilde{\bm{q}},\tilde{\bm{\iota}}]$, except for the third term in the second line, each term can be broken down w.r.t. the index of UAVs. Subsequently, for UAV $u_i$, we recall ${\zeta}_i,{\phi}_i,{\rho}_i,{\xi}_i,{\delta}_i,{\upsilon}_i,{\chi}_i,{\beta}_i,{\varepsilon}_i$ as the \textit{private (local) variables} and ${\lambda}_i$ as a \textit{public (global) variable}.

  Therefore, the term in the third line of~\eqref{eq:longlag1} makes our problem coupled and inseparable w.r.t. the UAV parameters and Lagrangian multipliers. To tackle this issue, we develope a consensus-based distributed algorithm, which consists of two steps to find the solution of the dual problem: i) solving the inner problem distributedly using distributed average consensus and the gradient descent method; ii) solving the outer problem distributedly using the consensus gradient method. We first treat each term of~\eqref{eq:longlag1} as a (hypothetically) separate term and rewrite the Lagrangian function as follows, in which the public variable ${\lambda}_j$ is replaced by a local variable   ${\lambda}^{(i)}_{j}$:
        \begin{equation}\label{eq:SepLag}
        \begin{aligned}
            L(\tilde{\bm{x}},\bm{\lambda},\bm{\zeta},&\bm{\phi},\bm{\rho},\bm{\xi},\bm{\delta},\bm{\upsilon},\bm{\chi},\bm{\beta},\bm{\varepsilon})\\&=\sum_{u_i\in \mathcal{U}}L_i(\tilde{\bm{x}},\bm{\lambda},\bm{\zeta},\bm{\phi},\bm{\rho},\bm{\xi},\bm{\delta},\bm{\upsilon},\bm{\chi},\bm{\beta},\bm{\varepsilon}),
            \end{aligned}
        \end{equation}
        where function $L_i$ is given in~\eqref{eq:longlag2}. 
        Our distributed algorithm  solves the problem through a series of GP approximations, which consists of two phases: i) for a given set of monomial approximations, obtaining the optimal solution; ii) using the obtained solution to derive the monomial approximations for the next round.  The first phase itself requires solving the outer and the inner problems iteratively using gradient-based methods. In the following, obtaining each optimal solution is considered as one ``iteration", while the iterations involved in solving the outer and the inner problem are called ``gradient-iteration".
 \subsubsection{Solving the inner problem distributedly using distributed average consensus and the gradient descent method} Considering Observation 1, given $h_j(\tilde{\bm{x}}^{[m-1]})$, the inner-problem can be written as the sum of separable terms w.r.t. the UAVs' indices. Consider $h_j(\tilde{\bm{x}}^{[m-1]})= \sum_{u_i\in \mathcal{U}}h^{(i)}_{j}(\tilde{\bm{x}}^{[m-1})$, where $h^{(i)}_{j}(\tilde{\bm{x}}^{[m-1]})= \sum_{v_k\in\mathcal{V}} \exp\left(\{\tilde{q}^{(i)}_{jk}\}^{[m-1]}+ \{\tilde{\iota}^{(i)}_j\}^{[m-1]}\right)$. Note that $h^{(i)}_{j}(\tilde{\bm{x}}^{[m-1]})$ can be computed locally at processor~$i$. Thus, $h_j(\tilde{\bm{x}}^{[m-1]})$, the sum of those values, can be obtained distributedly using a~\textit{distributed~average~consensus} method~\cite{xiao2007distributedCons1,distributedCons2,distributedCons3,distributedCons4,distributedCons5}. Afterward, the gradient descent method can be applied locally on each term of~\eqref{eq:longlag2}. Since at iteration $m$ the value of  $h_j(\tilde{\bm{x}}^{[m-1]})$ does not change through the gradient decent updates, i.e., the gradient-iterations, it needs to be calculated once prior to solving the inner and the outer problem, and thus the consensus method does not have a significant impact on the convergence speed since the convergence is usually achieved in a few number of iterations (e.g.,  $55$ in Section~\ref{dist:cons_conv}). 
 
  \subsubsection{Solving the outer problem distributedly using the consensus gradient method}
Considering Observation 2, our approach consists of two steps: i) updating the local variables at each processor, ii) updating the global variable. Each processor first locally derives the values of the local variables by applying the gradient ascent method on~\eqref{eq:SepLag}. For example, for $\rho_i$, at gradient-iteration $t+1$, processor $i$ performs as follows:
\begin{equation}
\begin{aligned}
   \rho_i^{[t+1]}=&\rho_i^{[t]}+ c_{\rho} \Big(\nabla_{\rho^i} D^i({\lambda^{(i)}_j}^{[t]},{\zeta_i}^{[t]},{\phi_i}^{[t]},{\rho_i}^{[t]},\\&{\xi_i}^{[t]},{\delta_i}^{[t]},{\upsilon_i}^{[t]},{\chi_i}^{[t]},{\beta_i}^{[t]},{\varepsilon_i}^{[t]})\Big),
   \end{aligned}
\end{equation}
where $c_{\rho}$ is the step size. Regarding the global variable, $\lambda_{j}$, $\forall v_j$, processor $i$ obtains a pseudo version of it ${\lambda'}^{(i)}_{j}$ as follows:
\begin{equation}
\begin{aligned}
  &{{\lambda'}^{(i)}_{j}}^{[t]}={\lambda^{(i)}_j}^{[t]}+ c \Big(\nabla_{\lambda^{(i)}_j} D^i({\lambda^{(i)}_j}^{[t]},\\&{\zeta_i}^{[t]},{\phi_i}^{[t]},{\rho_i}^{[t]},{\xi_i}^{[t]},{\delta_i}^{[t]},{\upsilon_i}^{[t]},{\chi_i}^{[t]},{\beta_i}^{[t]},{\varepsilon_i}^{[t]})\Big),
   \end{aligned}
\end{equation}
where, the local copies of the global variable ($\lambda^{(i)}_{j}$-s) are derived by employing the \textit{consensus gradient method}~\cite{ref:DistCens}:
\begin{align}\label{eq:full_dist_update}
{\lambda^{(i)}_{j}}^{[t+1]}=\sum_{m=1}^{|\mathcal{U}|} \Big(\mathbf{W}^\vartheta\Big)_{im} {\lambda'^{(m)}_{j}}^{[t]},
\end{align}
where $\mathbf{W}=\mathbf{I}-\epsilon \mathbf{L}(G_p)$, with $\mathbf{L}(G_p)$ the Laplacian matrix of the processors network graph $G_p$ and $\epsilon \in (0,1)$, and $\vartheta\in \mathbb{N}$ denotes the number of conducted consensus iterations among the adjacent processors. In this method, the adjacent processors perform $\vartheta$ consensus iterations by exchanging the local copies of ${{\lambda'}^{(i)}_{j}}$-s before updating ${{\lambda}^{(i)}_{j}}$. Due to the convexity of the Lagrangian function and the concavity of the dual function, the minimax theorem~\cite{fan1953minimax} holds for~\eqref{ref:dualProb} and thus the order of solving the inner and the outer problem can be interchanged. The pseudo-code of our distributed algorithm is given in Algorithm~\ref{alg:fulldist}. The convergence of our distributed algorithm to the KKT solution of~\eqref{eq:originalFormualtion} is the result of the convergence of the \textit{consensus gradient method}~\cite{ref:DistCens} along with the convergence of the proposed GP approximation method (see Proposition~\ref{propKKT2}).

\subsection{ Complexity Analysis and Comparison}

In our proposed centralized approach, all the computations are carried out in a central processor that has the entire knowledge of the UAVs' parameters and can solve the convex optimization problem proposed in step~\ref{Alg:Gpconvexste} of Algorithm~\ref{alg:cent}. This step is often implemented using a gradient descent technique even in commercial software. Considering the problem formulation in \eqref{eq:formulationAxiliaryRelaxed}, the dimension of the solution to the problem (including the auxiliary variables) is $|\mathcal{U}|\times \left( \underbrace{ |\mathcal{V}|^2+|\mathcal{V}|}_{(a)}+\underbrace{|\mathcal{V}|+1}_{(b)}\right)$, where the terms denoted by (a) correspond to the main solutions of the problem ($\bm{q}, \bm{\iota}$) and the terms denoted by (b) correspond to the auxiliary variables ($\bm{A},\bm{B}$). Further, the number of Lagrangian multipliers is $|\mathcal{V}| +|\mathcal{U}|\times \left(4+4 |\mathcal{V}|+|\mathcal{V}|^2\right)$. All of which need to be stored in the same processor and updated at the same time. 

In the proposed decentralized algorithm, we disperse the computations across multiple processors, each of which is associated with  $ |\mathcal{V}|^2+2|\mathcal{V}|+1$ primal variables and $4+5 |\mathcal{V}|+|\mathcal{V}|^2$ dual variables.\footnote{Note that the public variable $\lambda_{j}$, $\forall v_j$, is replicated at each processor.} This removes the scaling of the space complexity of the algorithm with the number of UAVs and allows parallel processing across multiple processors, both of which are highly desired upon existence of a large number of UAVs. Nevertheless, due to the existence of the public variable and the use of the consensus gradient method, messages of size $|\mathcal{V}|$ need to be exchanged among the processors at each iteration, introducing communication overhead. Note that the number of iterations required for the two algorithms to converge is similar since the distributed algorithm mimics the centralized algorithm. However, the convergence will roughly\footnote{If we ignore the time required for message exchange among the processors.} be $|\mathcal{U}|$ times faster in the distributed algorithm since all the processes are conducted in parallel among $|\mathcal{U}|$ processors. 
 In summary, the main differences between the proposed centralized and decentralized methods lie in memory usage, prallelization/speedup, and communication overhead.

        \section{Simulation Results}\label{sec:sim}
   
        \subsection{Simulation Setup}
     
      We consider $200$ realizations of a network graph consisting of $10$ sites; the distance between the sites is chosen uniformly at random between $5 \textrm{km}$ to $50\textrm{km}$.  The base node  of each UAV is chosen uniformly at random among the sites, and the desired inspection criteria of the nodes $\bm{\pi}$ is a randomly generated normalized vector.  We consider fixed wing UAVs moving with the average speed of $25\textrm{m/s}$, where $c_1=9.26\times 10^{-4}$ and $c_2=2.25 \times 10^3$ in~\eqref{eq:energyModel1}~\cite{zhangEnergy}. If a UAV decides to inspect a site, it slows down its movement speed to $12.5\textrm{m/s}$ to conduct the inspection. The duration of inspection of each site is chosen uniformly at random between $5 \textrm{min}$ to $25 \textrm{min}$ to obtain the energy of data collection.
      Modern UAVs can be equipped with hyperspacial sensors,  multi-spectral targeting systems (MTS), and light detection and ranging (LIDAR). We consider a basic application of data collection using imaging, where each UAV is equipped with a mini gyro stabilized EO/IR drone FLIR thermal imaging camera and an HD camera with power of $8\textrm{W}$ and $9\textrm{W}$, respectively. It is assumed that the UAVs have enough battery to fly between $60\textrm{km}$ to $120\textrm{km}$ when their sensors are turned off. In the following, all the figures represent the average performance over the $200$ realizations unless otherwise stated. Also, $\hat{\theta}_j=0.9$, $\forall v_j\in\mathcal{V}$ in~\eqref{ProbConst1}, and $\tilde{\theta}_i=0.7$, $\forall u_i\in\mathcal{U}$ in~\eqref{eq:battr}.  Since there is a lack of studies devoted to investigating the stochastic surveillance for energy limited random walkers with random inspection policies, we propose the following  stochastic surveillance baselines inspired by the Metropolis–Hastings (MH) Markov chain Monte-Carlo (MCMC) technique~\cite{MCMC1,MCMC2,MCMC3} combined with hard map partitioning:
      
      \textbf{1) MH-MCMC with Random Map Partitioning (MH\_RMP):} This baseline randomly partitions the set of network sites into disjoint subsets according to the number of UAVs, each of which contains roughly the same number of sites and only one base node. Each UAV inspects the nodes belonging to the same subset as its base. The UAVs sensors always turn on upon passing the sites. Given the normalized inspection criteria of the nodes inside each subset, the transition matrices of the Markov chains associated with the UAVs movement are obtained using the MH technique~\cite{MCMC1,MCMC2,MCMC3}.
      
         \textbf{2) MH-MCMC with Random Map Partitioning and Optimized Inspection Policies\\ (MH\_RMP\_OI):} It follows the same procedure as MH\_RMP except that it further optimizes the UAVs' inspection policies to reduce the energy consumption.
         
            \textbf{3) MH-MCMC with  Distance-based Map Partitioning  (MH\_DMP):} This baseline is similar to MH\_RMP with a different clustering rule. The nodes inside  each subset are chosen to be the closest nodes (corresponding to the least movement energy) to the respective base node. The rest of the procedure is the same as MH\_RMP.
            
             \textbf{4) MH-MCMC with  Distance-based Map Partitioning and Optimized Inspection Policies (MH\_DMP\_OI):} It follows the same procedure as MH\_DMP except that it further\\ optimizes the UAVs' inspection policies.
             
The results presented in Sections~\ref{subs:enery}, \ref{sec:simC} are obtained using the centralized algorithm, while the convergence of the distributed algorithm is studied in Section~\ref{dist:cons_conv}. 
     \begin{figure*}[t]
  \centering
\minipage{0.32\textwidth}
\vspace{-3mm}
  \includegraphics[width=2.1in, height=1.55in]{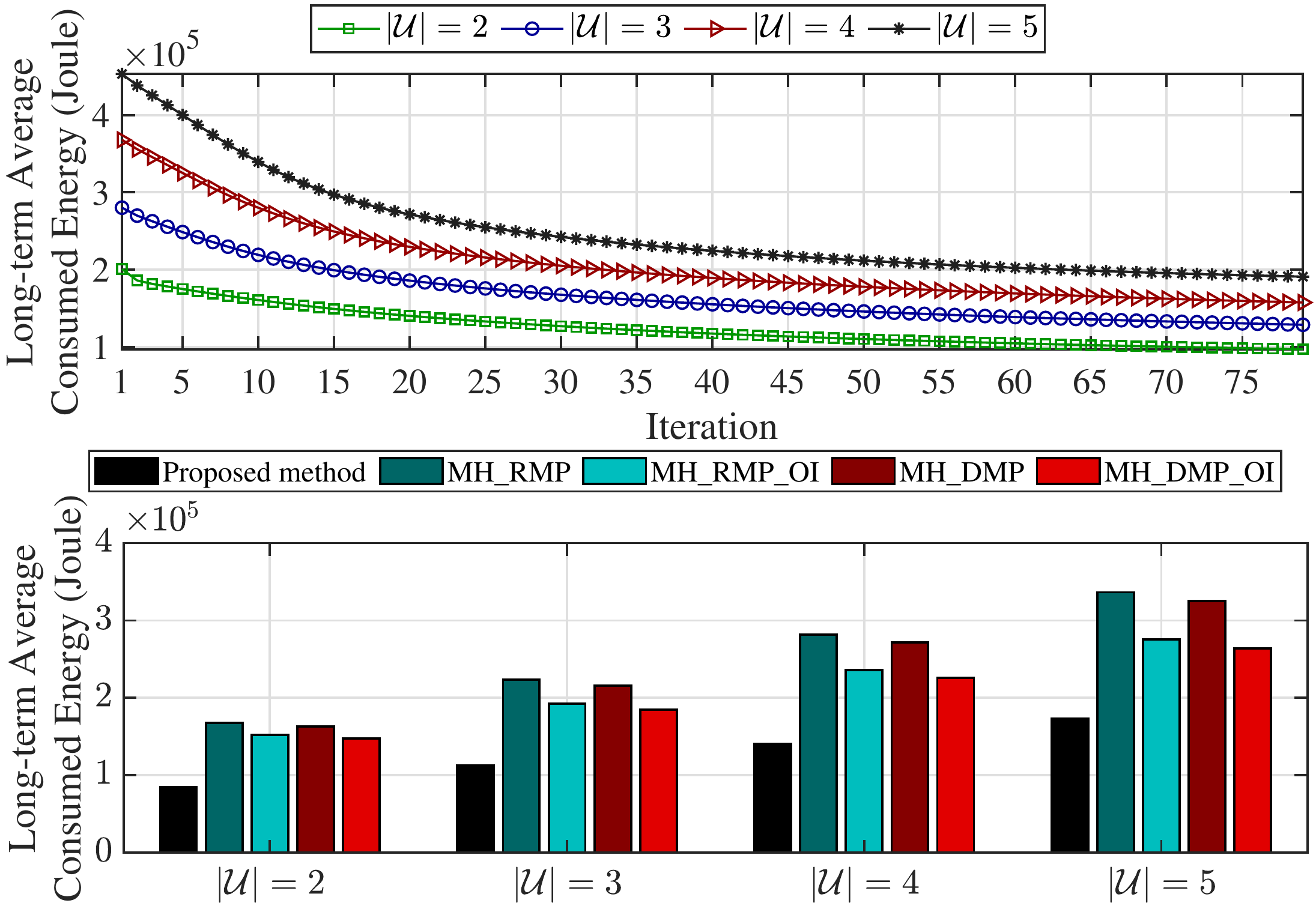}
  \caption{The long-term average consumed energy during the surveillance for different numbers of UAVs w.r.t. the iteration count of our algorithm (top plot). The corresponding comparison with the baselines (bottom plot).}\label{fig:sim1}
\endminipage\hfill
\minipage{0.32\textwidth}
\vspace{-3mm}
  \includegraphics[width=2.1in, height=1.55in]{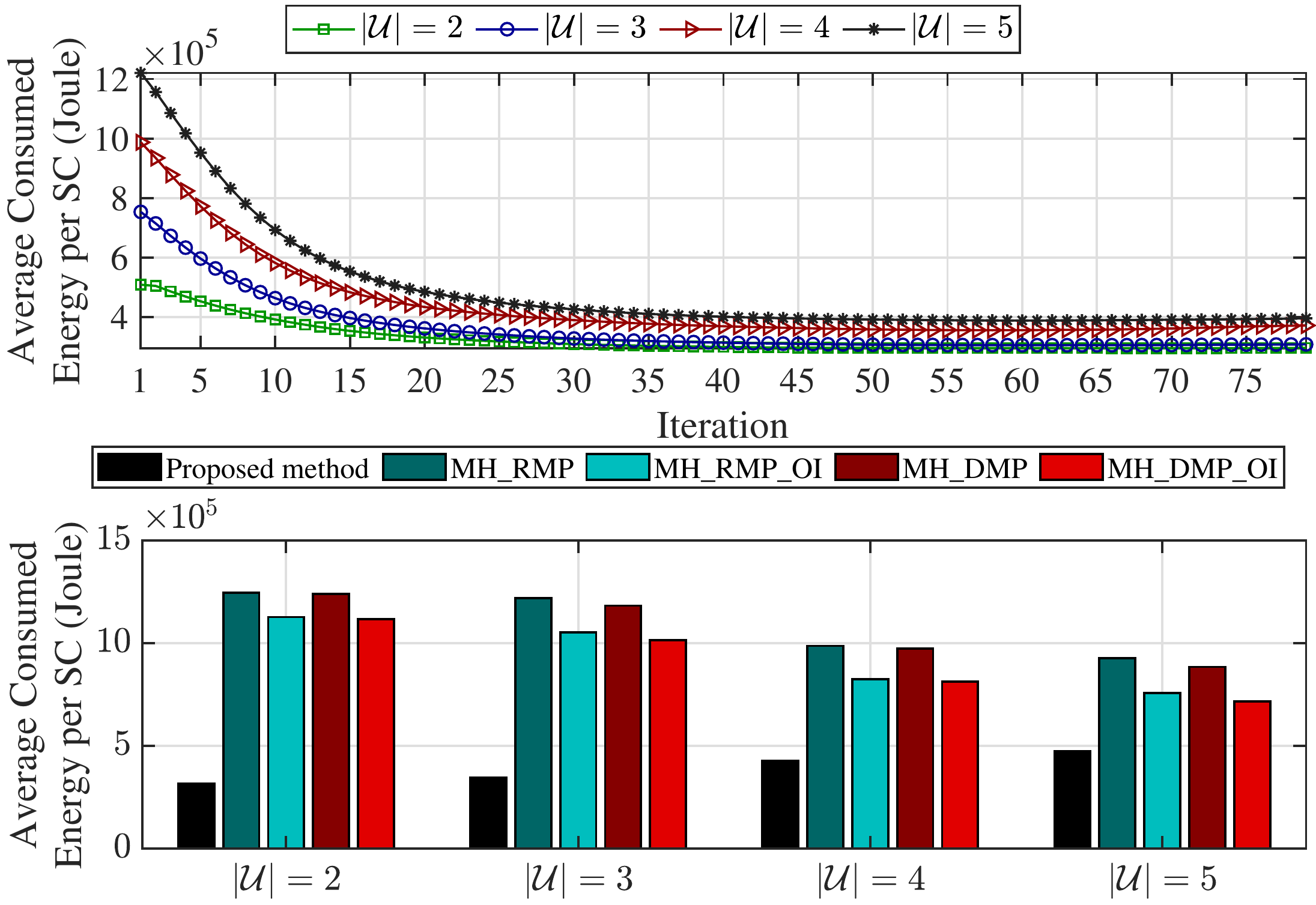}
  \caption{The average consumed energy per surveillance cycle (SC) for different numbers of UAVs w.r.t. the iteration count of our algorithm (top plot). The corresponding comparison with the baselines (bottom plot).}\label{fig:sim2}
\endminipage\hfill
\minipage{0.32\textwidth}
  \includegraphics[width=2.1in, height=1.55in]{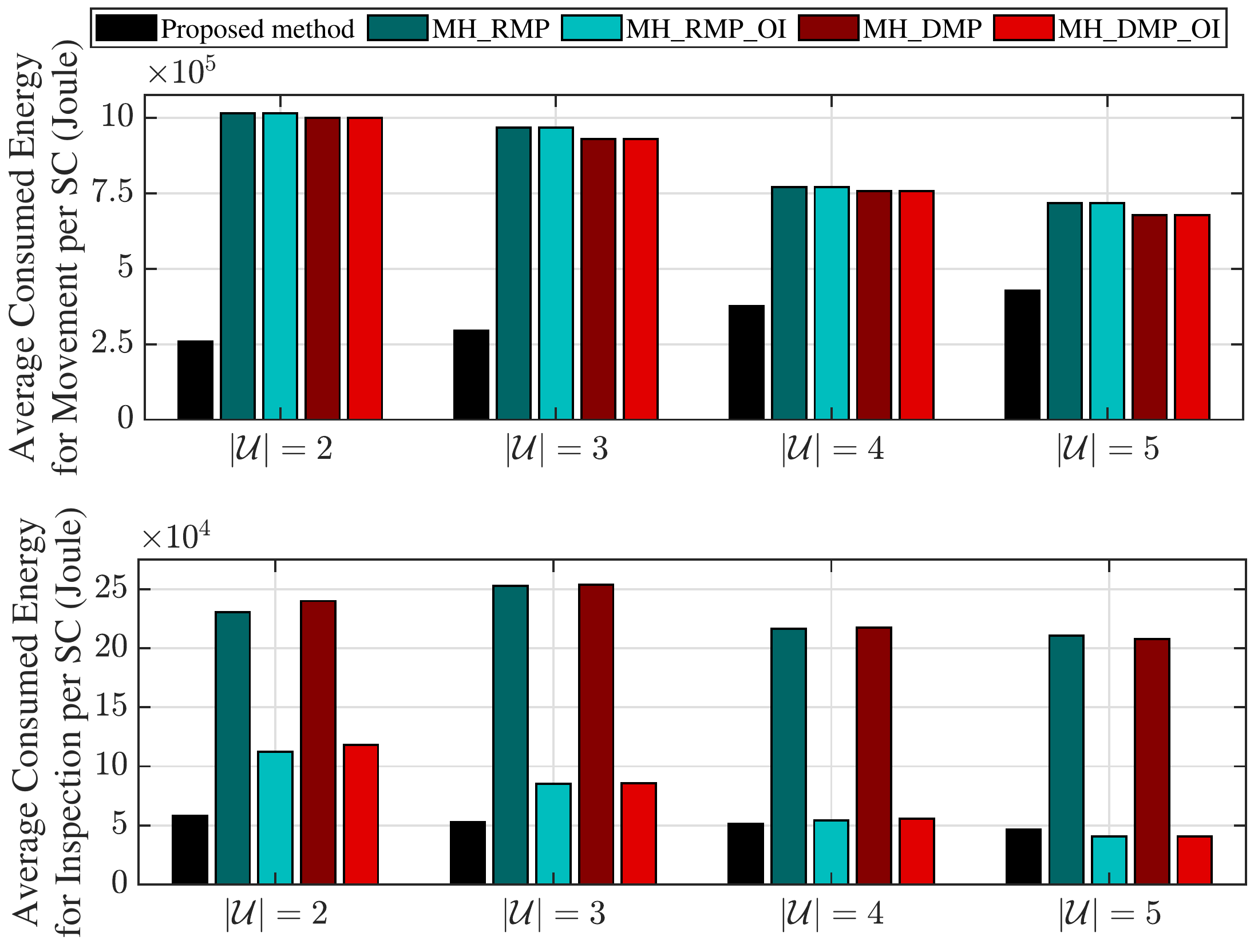}
  \caption{Top: The comparison of the average energy used for movement per surveillance cycle (SC) between our method and the baselines. Bottom: The comparison of the average energy used for inspection per SC between our method and the baselines. }\label{fig:sim3}
\endminipage
\end{figure*}
  \begin{figure*}[t]
  \centering
  \minipage{0.32\textwidth}
\vspace{-6.5mm}
 \includegraphics[width=2.1in, height=1.9in]{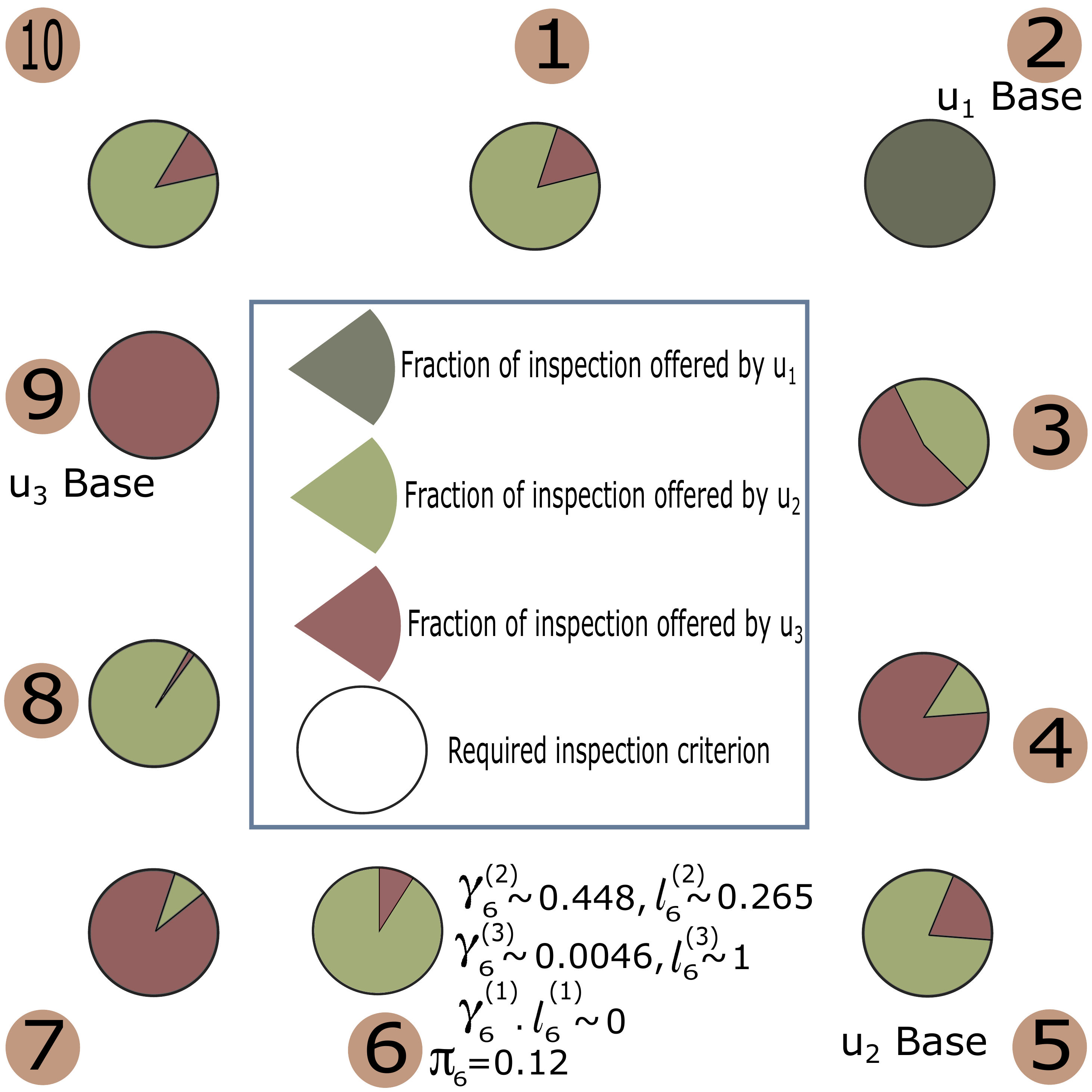}
	\caption{An example of soft map partitioning for a map of $10$ nodes upon having $3$ UAVs. The weights of the edges between the nodes are omitted for better readability.}
	\label{fig:sim9}
\endminipage\hfill
\minipage{0.32\textwidth}
  \includegraphics[width=\linewidth]{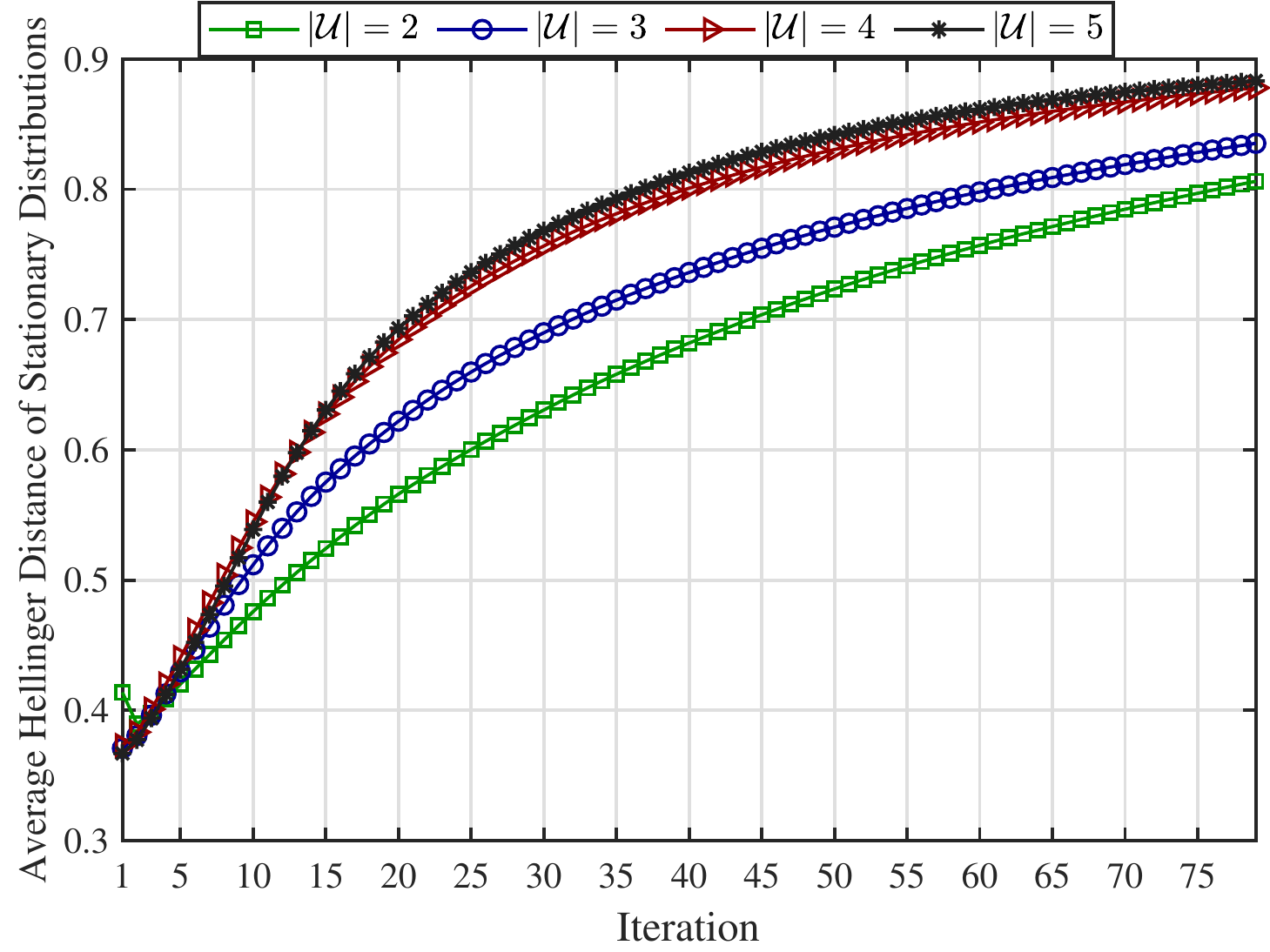}
  \caption{The (pairwise) average Hellinger distance between the stationary distributions of the UAVs for different numbers of UAVs.}\label{fig:sim4}
\endminipage\hfill
\minipage{0.32\textwidth}
  \includegraphics[width=\linewidth]{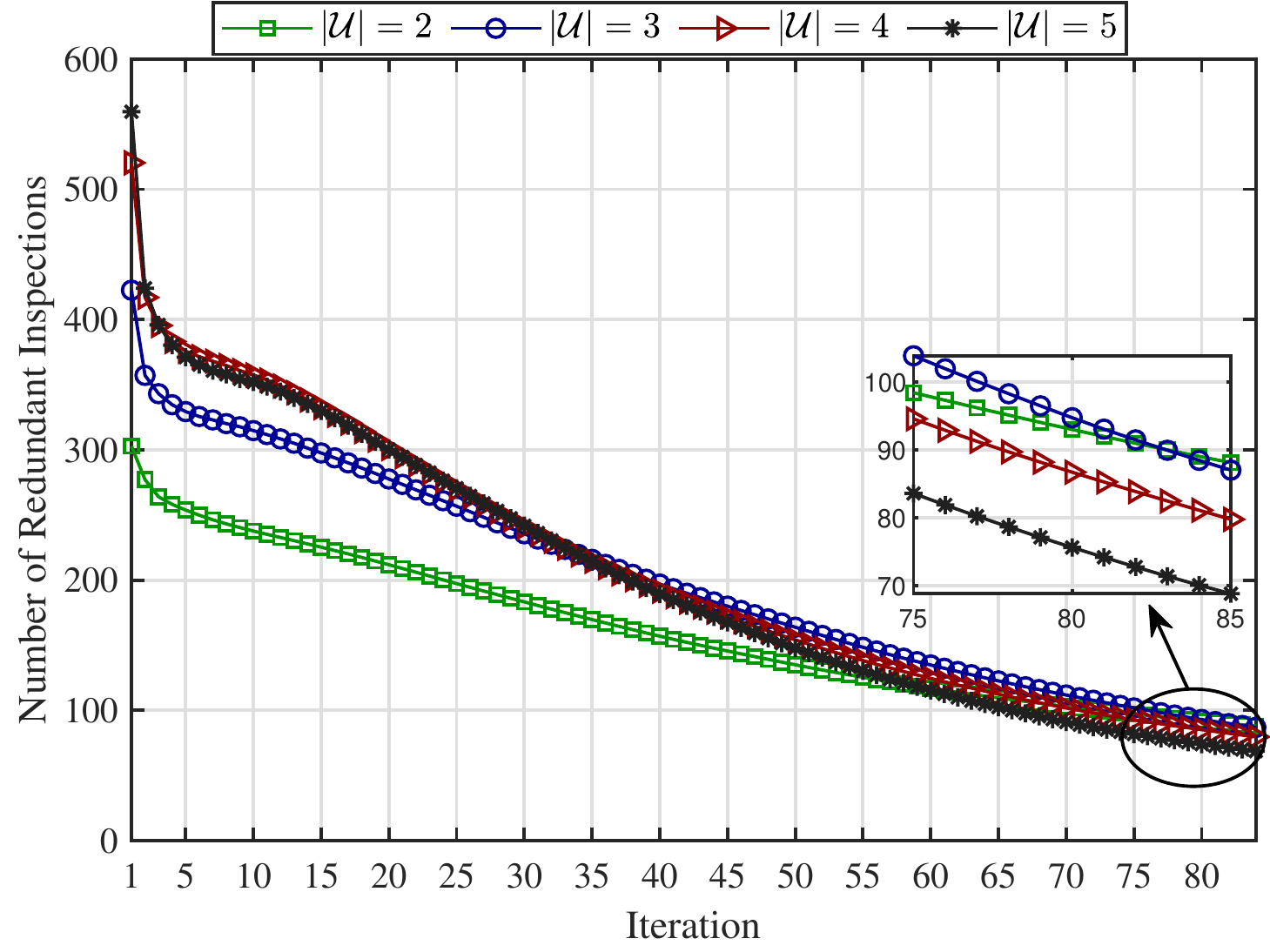}
  \caption{The number of redundant inspections upon letting the UAVs move for $10000$ time instances for different numbers of UAVs.}\label{fig:sim5}
\endminipage\hfill
\end{figure*}
    \begin{figure*}[h]
    \vspace{-9mm}
  \centering
  \minipage{0.32\textwidth}
    \vspace{11.1mm}
  \includegraphics[width=2.1in, height=1.55in]{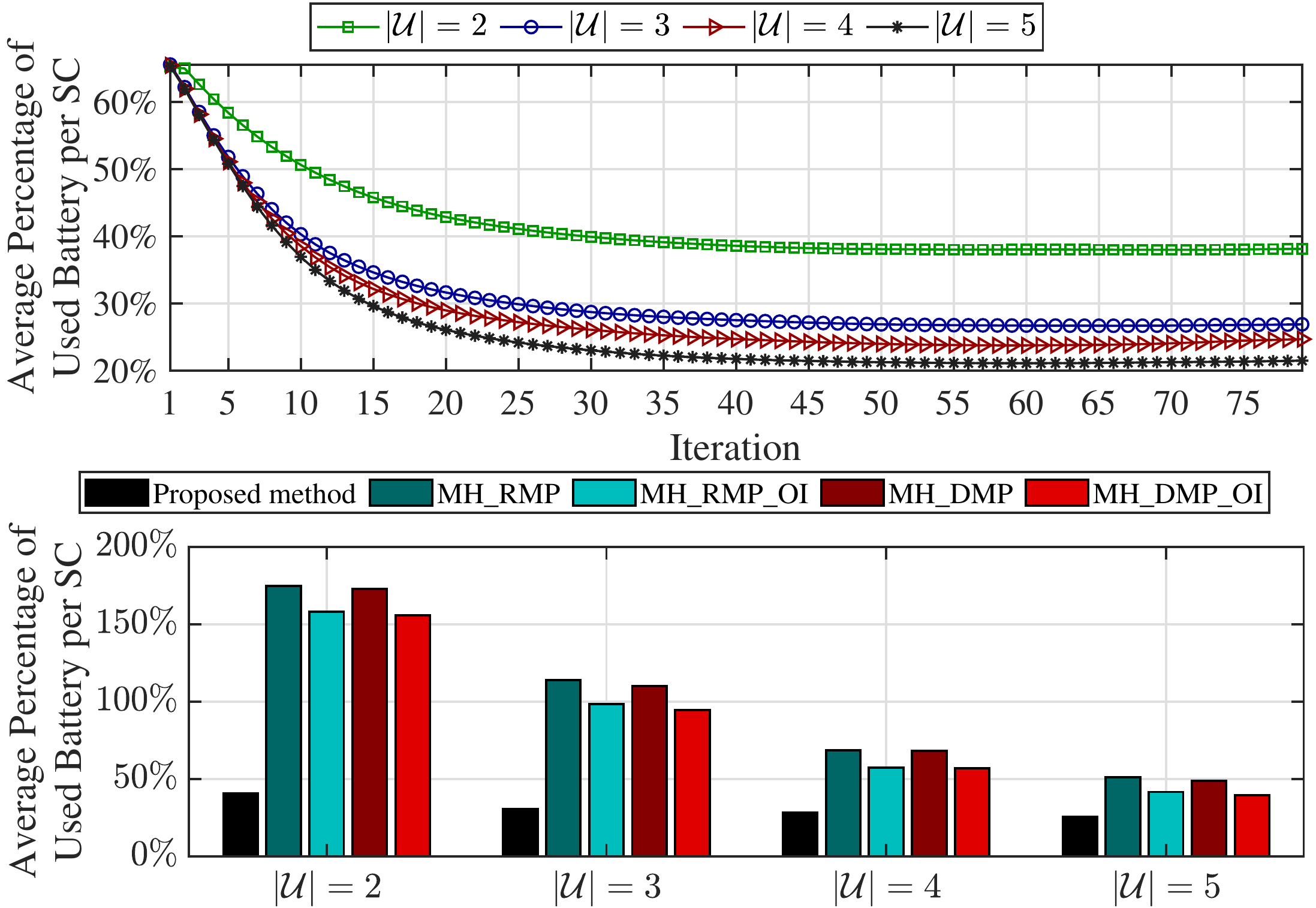}
  \caption{The average percentage of \hbox{used~battery} per surveillance cycle {\small(SC)}  for different numbers of UAVs \hspace{-.4mm}w.r.t. \hspace{-.65mm} the iteration count of our~algorithm (top). The corresponding comparison with the  baselines (bottom).}\label{fig:sim6}
\endminipage\hfill
\minipage{0.32\textwidth}
  \includegraphics[width=\textwidth]{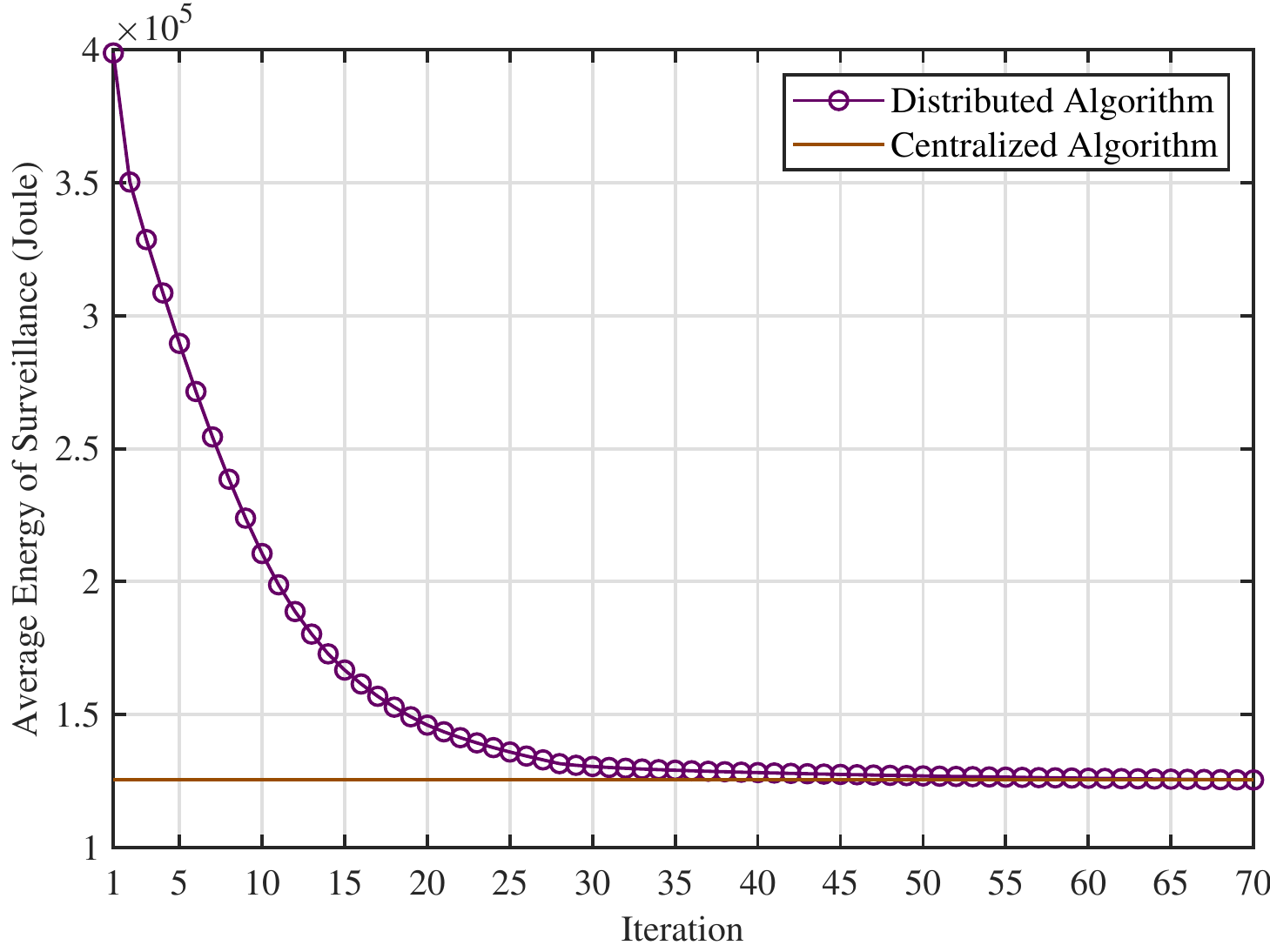}
	\caption{Convergence of the distributed algorithm considering the average cost of surveillance for $4$ UAVs.}
	\label{fig:sim7}
\endminipage\hfill
\minipage{0.32\textwidth}
  \vspace{3mm}
  	\includegraphics[width=\textwidth]{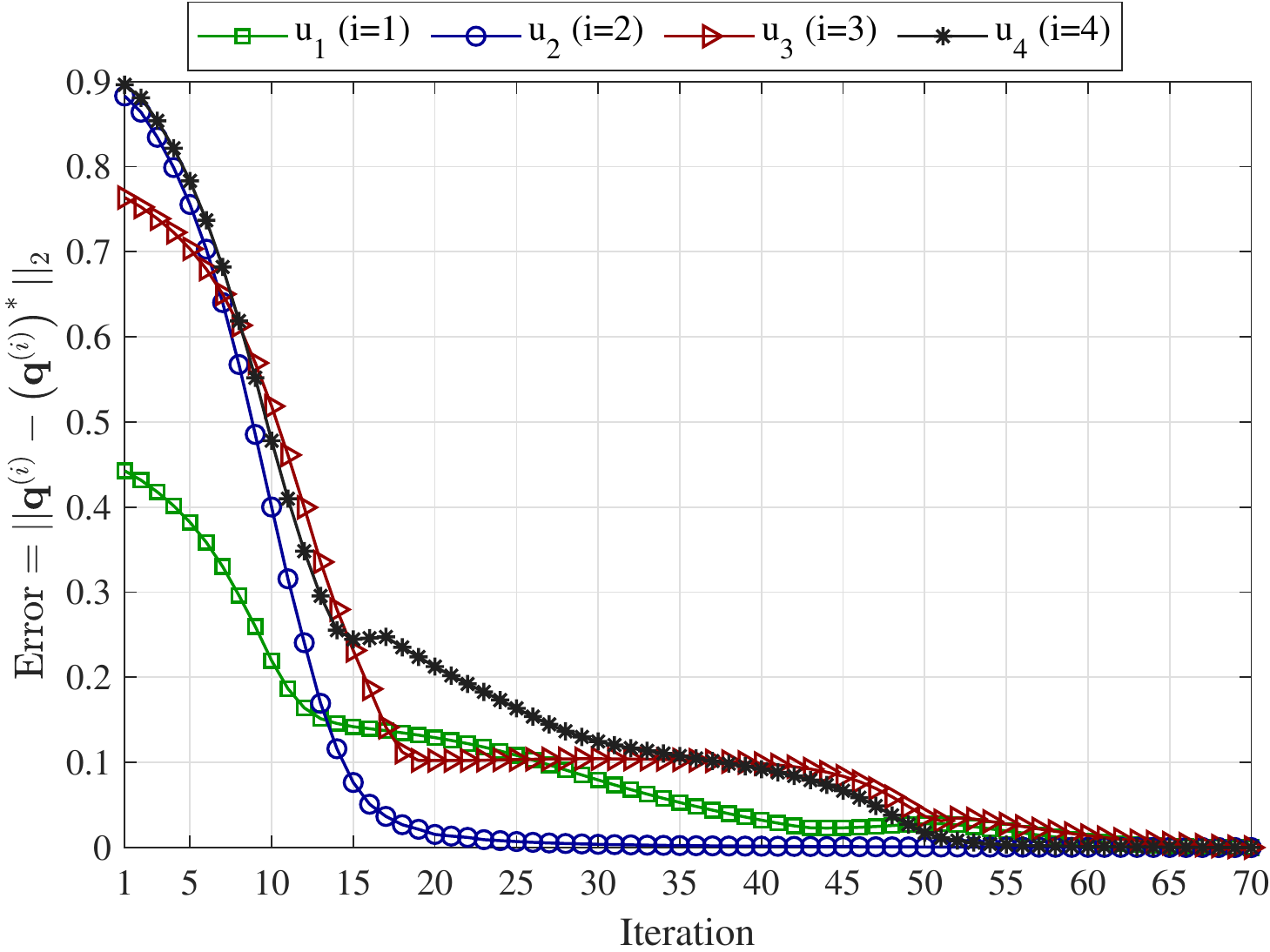}
	\caption{Error of convergence of the solution of the distributed algorithm $\mathbf{q}^{(i)}=[q^{(i)}]_{1\leq j,k\leq |\mathcal{V}|}$ to that of the centralized algorithm $\left(\mathbf{q}^{(i)}\right)^*$. }
	\label{fig:sim8}
\endminipage\hfill
\end{figure*}

  \subsection{Energy Efficiency}\label{subs:enery}
  In Fig.~\ref{fig:sim1}, the top plot depicts the value of the objective function of~\eqref{eq:originalFormualtion}, i.e., the long-term
average consumed energy during the surveillance, for different numbers of UAVs w.r.t. the iteration count; the bottom plot depicts the performance comparison between our method and the baseline methods. From the top plot, it can be seen that the objective function monotonically decreases through the series of monomial approximations upon convergence, thus verifying Proposition~\ref{propKKT2}. From the bottom plot, (on average) our method results in around $42\%$ energy saving as compared to the baseline methods. In Fig.~\ref{fig:sim2}, we generate similar plots to Fig.~\ref{fig:sim1} considering the
average consumed energy per surveillance cycle. Comparing the bottom plots of Figs.~\ref{fig:sim1},~\ref{fig:sim2}, the performance gap between our method and the baseline methods is even more prominent, especially with fewer number of UAVs, e.g., $|\mathcal{U}|=2$ and $|\mathcal{U}|=3$, upon considering the energy consumed per surveillance cycle. This is due to the underlying map partitioning approach utilized. In our solution, UAVs that need to reach the sites located far away from the rest are usually associated with a lower surveillance cycle duration, i.e., they visit fewer sites per surveillance cycle. For example, when $|\mathcal{U}|=2$ using the baseline methods (hard map partitioning), each UAV has to inspect $5$ nodes; however, this number can be different for different UAVs using our approach (see Section~\ref{sec:simC}). We break down the average consumed energy per surveillance cycle and depict the corresponding average consumed energy for movement and for inspection per surveillance in the top plot and the bottom plot of Fig.~\ref{fig:sim3}, respectively. A major performance gap can be seen in consumed energy for movement (top plot), which dominates the total energy consumption. Regarding the baseline methods, the distance-based partitioning of the map (MH\_DMP and MH\_DMP\_OI) yields better performance than random map partitioning. Also, the notable effect of optimized introspection policies can be seen from Fig.~\ref{fig:sim3} (bottom plot).
  \subsection{Soft Map Partitioning and Redundant Inspections}\label{sec:simC}
  As explained earlier, our algorithm leads to \textit{soft map partitioning}. This implies that the map will be probabilistically split among the UAVs, where each UAV will be present at certain partitions of the map with high probability (equivalently, most of the time). This has two main positive effects: i) it decreases the battery consumption of the UAVs; ii) it decreases the chance of \textit{redundant inspections}. The latter phenomenon happens when two UAVs inspect the same site at the same time, which leads to the wast of resources. As an example, Fig.~\ref{fig:sim9} depicts the final solution of our algorithm for one network realization upon having $3$ UAVs. Note that i) the topology of the graph and the edge weights are omitted for better readability, and thus the respective positions of the nodes do not convey any physical information, ii) the presented result is specific and may vary from one parameter setting to another. In this figure, it is illustrated that how the satisfaction of the nodes inspection criteria is achieved. As an example, the value of the inspection criterion of node 6 along with the obtained stationary distributions of the movements of the UAVs and their inspection policies are shown for node 6. It can be seen that UAV 1 will rarely leave its base node, i.e., node 2 (due to its low battery capacity and the distance of node 2 to its adjacent nodes, which are omitted for readability); the rest of the map is probabilistically partitioned among the other two UAVs to achieve the lowest average energy of surveillance. To mathematically quantify the map partitioning, we use the \textit{Hellinger distance}, which for two discrete distributions {\small$\bm{p}=(p_1,p_2,\cdots,p_n)$} and {\small$\bm{q}=(q_1,q_2,\cdots,q_n)$} is defined as: {\small$ H(\bm{p},\bm{q})=\frac{1}{\sqrt{2}} \sqrt{\sum_{i=1}^{n}\left(\sqrt{p_i}-\sqrt{q_i}\right)^2 }$}, also related to the total variation distance (or statistical distance) between the two distributions. Fig.~\ref{fig:sim4} depicts the pairwise average Hellinger distance between the stationary distributions of the movement of the UAVs w.r.t. the iteration count of our centralized algorithm. Furthermore, the number of redundant inspections upon letting the UAVs move for $10000$ time instances are depicted in Fig.~\ref{fig:sim5}. These two figures demonstrate that, as the number of iterations increases, our algorithm moves toward increasing the statistical distance between the UAVs' movements stationary distributions, i.e., splitting the map, and decreasing the number of redundant inspections. Also, from Fig.~\ref{fig:sim5} it can be noted that, initially a larger number of UAVs leads to a larger number of redundant inspections; however, when the
iteration count increases,
a smaller number of redundant inspections are incurred eventually  due to a sharper
map partitioning among the UAVs. Finally, in the top plot of Fig.~\ref{fig:sim6} we depict the average percentage of used battery of the UAVs during a surveillance cycle w.r.t. the iteration count of the algorithm; in the bottom plot the corresponding comparison with the baseline methods is depicted. From the top plot, it can be observed that after the first iteration the used battery ratios  are around the upper bound set by the value of $\tilde{\theta}_i=0.7$, $\forall u_i\in\mathcal{U}$; however, as the iteration count increases, the usage of batteries significantly drops. This illustrates the importance of map partitioning in decreasing the energy consumption. From the bottom plot of Fig.~\ref{fig:sim6}, we can see that the best baseline method is MH\_DMP\_OI, the battery usage of which is (on average) $50\%$ higher than that of our method. In some cases (upon having $|\mathcal{U}|=2$ and $|\mathcal{U}|=3$) the baseline methods result in more than $100\%$ average percentage of used battery per surveillance cycle. This usually implies failure of the UAVs in the return to their bases upon using those baseline methods.

  \subsection{Convergence of the Consensus-based Algorithm}\label{dist:cons_conv}
  We study the convergence of our proposed distributed algorithm assuming $4$ UAVs in the network.
 Considering the average energy of surveillance, for one network realization, Fig.~\ref{fig:sim7} depicts the convergence of our distributed algorithm. Furthermore, we depict the corresponding error of convergence in Fig.~\ref{fig:sim8}. As can be seen, although the distributed algorithm may start from a different initial point (a different set of initial monomial approximations of the posynomials), it eventually converges to the solution of the centralized algorithm. 
 
  \subsection{Key Trade-offs}

UAV-assisted networks are associated with multiple trade-offs (see~\cite{8641423} for the trade-offs concerned with the physical layer communications aspects). Our results and formulation also reveal some  tradeoffs for stochastic UAV-assisted surveillance that suggest interesting directions for future work:
\begin{enumerate} [leftmargin=5.5mm]
  \item Reliability vs. cost: Considering reliability in terms of mission accomplishment by the UAVs without running out of the batteries, higher reliability is achieved via increasing the number of deployed UAVs. This is because upon  increasing  the  number  of  UAVs  with  dispersed base nodes, it is easier to satisfy the desired inspection criterion of the sites while using less  battery  from  each  UAV,  since  the  UAVs  will  mostly  hover  around  their  base  nodes. However, this is usually accompanied by a larger total energy consumption and increased network cost. 
     \item Redundancy vs. predictability: Our approach naturally favors reducing the redundancy since it aims to decrease the amount of overlaps between the UAVs inspections (i.e., it tries to avoid having multiple UAVs inspecting the same site at the same time) to decrease the total energy consumption during the surveillance. Nevertheless, decreasing the redundancy in turn increases the predictability for malicious users. This is because, in the low redundancy regime, if a malicious user can detect/observe one UAV in a particular region of a site, it can make sure that no other UAVs is inspecting other regions of the site. 
     \item Feasibility vs. tolerance: The feasibility of the problem is dependent on battery limitations of the UAVs and the tolerance variables $\hat{\theta}_j,~\forall v_j\in \mathcal{V}$ and $\tilde{\theta}_i,~\forall u_i\in\mathcal{U}$ in (1) and (2). More tolerance on violating the probabilistic constraints (1) and (2) can result in larger feasibility regions for the problem given UAVs' limited battery capacities, and thus existence of a  solution. Therefore, although the network operator may originally desire low tolerance on violation of (1) and (2), the feasibility of the problem should also be taken into consideration. Thus, given the consideration that conducting a surveillance under high tolerance against errors is usually better than conducting no surveillance, the network operator needs to tune the aforementioned two tolerance parameters to ensure the existence of a solution while having the smallest deviation from the original desired tolerance.
\end{enumerate}

 \section{Conclusion and Future Work}\label{sec:conc}
  \noindent We proposed a novel framework for UAV-assisted surveillance utilizing random walks that inherently considers the battery constraints of the UAVs. We also introduced another degree of randomness to the system, which is the probabilistic inspection of the sites. We formulated the problem of jointly optimizing the random walk patterns and inspection policies of the UAVs, which turned out to be signomial programming. To tackle the problem, we proposed an iterative geometric programming approximation of the problem, and prove its optimally. We also took one step further and developed a distributed algorithm for the problem along with its performance guarantee. 
For the future work, formulating and investigating the problems explicitly considering the communications between the UAVs and terrestrial base nodes/stations is particularly interesting for  civil and commercial applications. In particular, the stochastic UAV-assisted inspection problem under the following two conditions can be investigated: (i) constant communication requirement between the UAVs to a specified set of  base nodes, and (ii) periodic communication and content delivery  from the UAVs to a set of trusted base nodes or cellular base stations. Both problems can be further investigated when the UAVs engage in a cooperative framework with data exchange among themselves over the air.
 \appendices
       \vspace{-4mm}
 \section{UAV Energy Consumption Models}\label{Sec:costModel}
  \vspace{-2mm}
\noindent We present compact energy models for the fixed and rotary wing UAVs, which are easy to use in practice. 
We use an energy model inspired by~\cite{zhangEnergy,zeng2019energyRotary,filippone2006flight}. Assume that at time $t=0$, a UAV starts traveling from site $v_j$ to $v_i$, which takes $T_{ji}\in \mathbb{R}^+$ amount of time. Let $y_{ji}$ denote the corresponding physical trajectory, where $y_{ji}(t)$ is the location of the UAV at time $t\in [0,T_{ji}]$.  For fixed-wing UAVs, the total propulsion energy is given by (see~\cite{zhangEnergy}, Appendix A):
\vspace{-1mm}

{\small
\begin{equation}\label{eq:energyModel1}
\begin{aligned}
     &w(v_j,v_i)=\int_{0}^{T_{ji}} \Bigg[c_1 \norm{\dot{y}(t)}^3\\
    &+\frac{c_2}{\norm{\dot{y}(t)}}\left(1+\frac{\norm{\ddot{y}(t)}^2-\frac{\left([\ddot{y}(t)]^{\top}\dot{y}(t)\right)^2}{\norm{\dot{y}(t)}^2}}{g^2}\right)\Bigg] dt \\
    &+ \frac{1}{2} m\left({\norm{\dot{y}(T_{ji})}}^2-\norm{\dot{y}(0)}^2\right),
    \end{aligned}
\end{equation}}
\vspace{-1mm}

\noindent where $\dot{y}(t)$ and $\ddot{y}(t)$ denote the velocity and acceleration vector, respectively, $g=9.8~m/s^2$, and $m$ is the mass of the UAV. Also, $c_1=\frac{1}{2}\rho C_{D0}S$ and $c_2=\frac{2W^2}{\pi e_0 A_R \rho S}$ are two constants, where $\rho$ is the air density in $kg/m^3$, $C_{D0}$ is the zero-lift drag coefficient of the UAV, $S$ is a reference area (e.g.,  the  wing  area), $W$ is the UAV weight in Newton, $e_0$ is  the  Oswald  efficiency (typically between $0.7$ and $0.85$), and $A_R$ is the aspect ratio of the wing, i.e., the ratio of the wing span to  its  aerodynamic  breadth.

For the rotary wing UAVs, in general the derivations are more complicated. Ignoring the acceleration of the UAV, the total propulsion energy is given by (see~\cite{zeng2019energyRotary}, Appendix):
\vspace{-1mm}

{\small
\begin{equation}
\begin{aligned}
    &w(v_j,v_i)=\int_{0}^{T_{ji}}\Bigg[P_0\left(1+\frac{3\norm{\dot{y}(t)}^2}{U^2_{tip}}\right)\\
    &+P_1\sqrt{\sqrt{1+\norm{\dot{y}(t)}^4/(4v^4_0)}-\norm{\dot{y}(t)}^2/(2v^2_0)}\\
    &+\frac{1}{2}d_0\rho s A \norm{\dot{y}(t)}^3 \Bigg]dt ,
    \end{aligned}
\end{equation}
}
\vspace{-1mm}

\noindent where the first, the second and the third terms inside the integral represent the blade profile power needed to overcome the profile drag of the blades, induced energy required to overcome the induced drag of the
blades, and parasite power needed to overcome the
fuselage drag, respectively. Also, $U_{tip}$ is the tip speed of the rotor balde, $v_0$ is the mean rotor induced velocity, $d_0$ is the fuselage drag ratio, $s$ is the rotor solidity, $\rho$ denotes the air density, and $A$ denotes rotor disc area in $m^2$. Furthermore, $P_0=\delta \rho s A \Omega^3 R^3/8$ and $P_1=(1+k)W^{3/2}/\sqrt{2\rho A}$, where $\delta$ is the profile drag coefficient, $\Omega$ is the blade angular velocity in $rad/sec$, $R$ is the rotor radius, $k$ is the incremental correction factor to induced power, and the rest of notations are similar to those in~\eqref{eq:energyModel1}. These physical layer expressions are used to derive the weight of the edges of the network. 
      \vspace{-2mm} 
 \section{Proof of Theorem~\ref{th:main}}\label{App:ProbtoTrac}
 \noindent Considering \eqref{eq:fo2}, to obtain a tractable expression,
define the Bernoulli random variable $\hat{I}^{(i)}_{j}\triangleq 1-I^{(i)}_{j}$. Considering the left hand side (l.h.s) of inequality \eqref{eq:fo2}, we get:
\vspace{-1mm}

{\small
  \begin{equation}
  \begin{aligned}
 &\textrm{Pr} \left( \sum_{i=1}^{|\mathcal{U}|} \gamma^{(i)}_j I^{(i)}_j \leq \pi_j\right)\\
 &=\textrm{Pr} \left( \sum_{i=1}^{|\mathcal{U}|} \gamma^{(i)}_j \hat{I}^{(i)}_{j} \geq  \sum_{i=1}^{|\mathcal{U}|} \gamma^{(i)}_j - \pi_j\right)\\
 &\leq \frac{E\left[\sum_{i=1}^{|\mathcal{U}|}\gamma^{(i)}_j \hat{I}^{(i)}_{j} \right]}{\sum_{i=1}^{|\mathcal{U}|}  \gamma^{(i)}_j - \pi_j}=\frac{\sum_{i=1}^{|\mathcal{U}|} \gamma^{(i)}_j\left(1-\iota^{(i)}_{j}\right)}{\sum_{i=1}^{|\mathcal{U}|}  \gamma^{(i)}_j - \pi_j},
 \end{aligned}
\end{equation}}
\vspace{-1mm}

\noindent where the inequality is the result of the {Markov inequality}. Note that $\sum_{i=1}^{|\mathcal{U}|} \gamma^{(i)}_j \geq \pi_j$, $\forall v_j\in \mathcal{V}$, is implicitly assumed, and in fact it will
be satisfied in the final solution; since otherwise the inspection criteria of the nodes cannot be satisfied even if all the UAVs turn on their sensing devises all the time. Using the above equation, the result of Theorem~\ref{th:main} can be obtained. 

        \begin{table*}[t]
\begin{minipage}{0.99\textwidth}
{\footnotesize{
      \begin{equation}\label{eq:longPrp1}
          \hspace{-5mm}
          \begin{aligned}
               \hat{h}_j(\bm{x}^{[m]})= \prod_{u_i\in \mathcal{U}} \prod_{v_k\in\mathcal{V}}  \left(\frac{\{q^{(i)}_{jk}\}^{[m]} \{\iota^{(i)}_{j}\}^{[m]} h_j(\bm{x}^{[m]})}{\{q^{(i)}_{jk}\}^{[m]}\{\iota^{(i)}_{j}\}^{[m]}}) \right)^\frac{\{q^{(i)}_{jk}\}^{[m]}\{\iota^{(i)}_{j}\}^{[m]}}{h_j(\bm{x}^{[m]})}= h_j(\bm{x}^{[m]})^\frac{\sum_{u_i\in \mathcal{U}} \sum_{v_k\in\mathcal{V}} \{q^{(i)}_{jk}\}^{[m]}\{\iota^{(i)}_{j}\}^{[m]}}{h_j(\bm{x}^{[m]})}={h}_j(\bm{x}^{[m]})
                    \end{aligned}
          \end{equation}
           \hrulefill   \begin{equation}\label{eq:longPrp2}
          \hspace{-4mm}
          \begin{aligned}
            &\frac{\partial \left(  \frac{j(\bm{x})}{\hat{g}(\bm{x})} \right)}{\partial x_i} \Bigg|_{\bm{x}=\bm{x}^{[m]}}\hspace{-8mm}=\frac{\frac{\partial j(\bm{x})}{\partial x_i}\hat{g}(\bm{x})\hspace{-.8mm} - \hspace{-.8mm} \frac{\partial \hat{g}(\bm{x})}{\partial x_i} j(\bm{x})}{\left(\hat{g}(\bm{x})\right)^2}\Bigg|_{\bm{x}=\bm{x}^{[m]}}\hspace{-7mm}\overset{g(\bm{x}^{[m]})= \hat{g}(\bm{x}^{[m]})}{=}\hspace{-0mm}\frac{\frac{\partial j(\bm{x})}{\partial x_i}g(\bm{x})\hspace{-.8mm} -\hspace{-.8mm} \frac{\hspace{-1.8mm}\partial\left(\displaystyle \prod_{k=1}^{K}\left( \frac{u_k(\bm{x})}{\alpha_k(\bm{x}^{[m]})}\right)^{\alpha_k(\bm{x}^{[m]})}\right)}{\partial x_i} j(\bm{x})}{\left(g(\bm{x})\right)^2}\Bigg|_{\bm{x}=\bm{x}^{[m]}}
            \\&
            =\frac{\frac{\partial j(\bm{x})}{\partial x_i}g(\bm{x}) - \frac{\displaystyle\sum_{n=1}^{K}\alpha_n(\bm{x}^{[m]})\frac{1}{\alpha_n(\bm{x}^{[m]})} \frac{\partial u_n(\bm{x})}{\partial x_i} \left( \frac{u_n(\bm{x})}{\alpha_n(\bm{x}^{[m]})}\right)^{\alpha_n(\bm{x}^{[m]})-1} \left(\displaystyle\prod_{k=1, k\neq n}^{K}\left( \frac{u_k(\bm{x})}{\alpha_k(\bm{x}^{[m]})}\right)^{\alpha_k(\bm{x}^{[m]})}\right)}{\partial x_i} j(\bm{x})}{\left(g(\bm{x})\right)^2}\Bigg|_{\bm{x}=\bm{x}^{[m]}}  \\[-.5em]&=  \frac{\frac{\partial j(\bm{x})}{\partial x_i}g(\bm{x}) - \frac{\displaystyle\sum_{n=1}^{K}\frac{\partial u_n(\bm{x})}{\partial x_i} g(\bm{x})^{\alpha_n(\bm{x})-1} \left(g(\bm{x})^{\sum_{k=1, k\neq n}^{K}\alpha_k(\bm{x}^{[m]})}\right)}{\partial x_i} j(\bm{x})}{\left(g(\bm{x})\right)^2}\Bigg|_{\bm{x}=\bm{x}^{[m]}}\hspace{-8mm}= \frac{\frac{\partial j(\bm{x})}{\partial x_i}g(\bm{x}) - \displaystyle\sum_{n=1}^{K}\frac{\partial u_n(\bm{x})}{\partial x_i} j(\bm{x})}{\left(g(\bm{x})\right)^2}\Bigg|_{\bm{x}=\bm{x}^{[m]}}\hspace{-8mm}=\frac{\partial \left(  \frac{j(\bm{x})}{g(\bm{x})} \right)}{\partial x_i}
                    \end{aligned}
                    \hspace{-6mm}
          \end{equation}
          }}
          \hrulefill
\end{minipage}
\vspace{-5mm}
\end{table*} 
  To derive a tractable expression for \eqref{eq:for4}, we use the following lemma and the result of the \textit{renewal reward theorem}.
  \vspace{-1mm}
  \begin{lemma}[\textbf{Mean return time}]\label{lemma:numberOfVisits}
Consider UAV $u_i$  with return time $T^{(i)}_{+}$. Given that the UAV starts the surveillance from its base node, i.e., $X_i(0)=v_{b^i}$, we have~\cite{aldous}:
 $E_i[T^{(i)}_{+}]={1}/{\gamma^{(i)}_{b^i}}.$
  \end{lemma}
   \vspace{-1mm}
  \begin{definition}[\textbf{Reward process}]
  Consider a counting process $(N(t) : t \geq 0)$ associated with i.i.d. inter renewal times $(X_n : n \in \mathbb{N})$ having common
distribution $F$. At the end of the $n$-th renewal interval, a random reward $R_n$ is earned. Let $(X_n, R_n)$ be
i.i.d. with the reward $R_n$ possibly dependent on $X_n$. Then the reward process
$(R(t) : t \geq 0)$ consists of accumulated reward earned by time $t$ as $R(t) = \sum_{i=1}^{N(t)}  R_i$.
  \end{definition}
   \vspace{-1mm}
  \begin{theorem}[\textbf{Renewal reward theorem} \hspace{-0.1mm}\cite{doob1948renewal,cinlar1969markov}]\label{th:rewardThreom}
Let N(t) be a counting process associated with $(X_n,R_n)$, $n \geq 1$.
Assuming $r = E[R_1]< \infty$ and  $\tau= E[X_1] < \infty$, we have: 
\begin{equation}
    \lim_{ t \rightarrow \infty} \frac{E[R(t)]}{t}=\frac{r}{\tau}.
\end{equation}
  \end{theorem}
  Let us define $ E[M_i]\overset{\Delta}{=}E\left[\sum_{t=1}^{T^{(i)}_{+}}w(X_i(t),X_i(t+1)) \right]$ and $E[L_i]\overset{\Delta}{=} E\left[\sum_{t=1}^{T^{(i)}_{+}}\psi^{(i)}_{X_i(t)}I^{(i)}_{X_i(t)}\right]$. In other words, $E[M_i]$ and $E[L_i]$ refer to the \textit{expected value of the movement energy} and the \textit{expected value of the inspection energy} per surveillance cycle, respectively. In the following, we derive a closed-form expression for each of them in order.
  For UAV $u_i$, the long term average expected energy of movement is given by:
  \vspace{-1mm}
  
   {\footnotesize
        \begin{equation}\label{eq:midS1}
         \hspace{-19mm}  \lim_{T\longrightarrow \infty}\frac{1}{T} E\bigg[ \sum_{t=1}^{T} w(X_i(t),X_i(t+1)\bigg]\hspace{-1mm}=\hspace{-1mm} \sum_{v_j\in\mathcal{V}}\sum_{v_k\in\mathcal{V}} \gamma^{(i)}_j  p^{(i)}_{jk}w(v_j,v_k).\hspace{-12mm}
        \end{equation}
        }\vspace{-1mm}
        
        \noindent On the other hand, we can obtain the following expression:
        \vspace{-1mm}
        
        {\small
        \begin{equation}\label{eq:midS2}
           \hspace{-21mm}  \lim_{T\longrightarrow \infty} \frac{1}{T} E\bigg[ \sum_{t=1}^{T} w(X_i(t),X_i(t+1)\bigg]= \frac{E[M_i]}{E [{T^{(i)}_{+}}]}
            =\frac{E[M_i]}{1/\gamma^{(i)}_{b^i}},\hspace{-12mm}
        \end{equation}
        }\vspace{-1mm}
        
       \noindent where the first and the second equality are the result of Theorem~\ref{th:rewardThreom} and Lemma~\ref{lemma:numberOfVisits}, respectively.
        Comparing~\eqref{eq:midS1} with \eqref{eq:midS2}, we get:
        \begin{equation}\label{eq:EM}
          E[M_i]=\frac{1}{\gamma^{(i)}_{b^i}}\sum_{v_j} \gamma^{(i)}_j \sum_{v_k} p^{(i)}_{jk}w(v_j,v_k).
        \end{equation}

  Also, considering UAV $u_i$, the long term average expected energy of nodes inspections can be expressed as:
 \begin{equation}
      \lim_{T\longrightarrow \infty} \frac{1}{T} E\bigg[ \sum_{t=1}^{T}\psi^{(i)}_{X_i(t)}I^{(i)}_{X_i(t)}\bigg]={\gamma^{(i)}_j} \psi^{(i)}_{j} \iota^{(i)}_{j}.
 \end{equation}
Using Theorem~\ref{th:rewardThreom} and Lemma~Lemma~\ref{lemma:numberOfVisits}, we get:
 \begin{equation}
     \lim_{T\longrightarrow \infty} \frac{1}{T} E\bigg[ \sum_{t=1}^{T}\psi^{(i)}_{X_i(t)}I^{(i)}_{X_i(t)}\bigg]=\frac{E[L_i]}{E [{T^{(i)}_{+}}]}
            =\frac{E[L_i]}{1/\gamma^{(i)}_{b^i}},
 \end{equation}
Comparing the above two equations, we get:
\begin{equation}\label{eq:EL}
     E[L_i]=\sum_{v_j\in\mathcal{V} } \frac{\gamma^{(i)}_j}{\gamma^{(i)}_{b^i}} \psi^{(i)}_{j} \iota^{(i)}_{j}.
\end{equation}
    Considering the l.h.s. of~\eqref{eq:for4}, using \textit{Markov inequality}, we get:
    \vspace{-1mm}
    \begin{equation}
          \begin{aligned}
     &\textrm{Pr}\bigg(\sum_{t=1}^{T^{(i)}_{+}}w(X_i(t),X_i(t+1)) +\\
     &\sum_{t=1}^{T^{(i)}_{+}}\psi^{(i)}_{X_i(t)}I^{(i)}_{X_i(t)} \geq \varphi_i \bigg)\leq \frac{E[M_i]+ E[L_i]}{\varphi_i},
    \end{aligned}
  \end{equation}
 After replacing the results of~\eqref{eq:EM} and~\eqref{eq:EL} in the above equation, the  result of Theorem~\ref{th:main} can be obtained.
 

\vspace{-2mm}
\section{Proof of Proposition~\ref{propKKT2}}\label{app:1}
   \noindent  We first prove that algorithm~\ref{alg:cent} generates a sequence of improved
feasible solutions that converge to a point $\bm{x}^*$
satisfying the KKT conditions of~\eqref{eq:formulationAxiliary}. Note that ~\eqref{eq:formulationAxiliaryRelaxed} is in fact an \textit{inner approximation} of~\eqref{eq:formulationAxiliary}~\cite{GeneralInnerApp}. Hence, it is sufficient to prove the following three characteristics for~\eqref{eq:formulationAxiliaryRelaxed}~\cite{GeneralInnerApp}:

     \textbf{1) All the approximations conducted in~\eqref{eq:formulationAxiliaryRelaxed} should result in tightening the constraints in~\eqref{eq:formulationAxiliary}; mathematically, for every inequality in the original problem~\eqref{eq:formulationAxiliary} in the form of $v(\bm{x})\leq1$ and its approximated version $\hat{v}(\bm{x})\leq1$ in~\eqref{eq:formulationAxiliaryRelaxed}, we should have $v(\bm{x})\leq \hat{v}(\bm{x})$}. Considering $\mathbf{\tilde{C}1}$ and $\mathbf{\hat{C}1}$ as an example, we get:
     \vspace{-1mm}
     
          {\small
          \begin{equation}
          \hspace{-5mm}
          \begin{aligned}
                 &\hat{h}_j(\bm{x})\geq  {h}_j(\bm{x}) \Rightarrow  \frac{\hat{\theta}_j \pi_j+\left(1-\hat{\theta}_j\right)\displaystyle\sum_{u_i\in \mathcal{U}} \sum_{v_k\in\mathcal{V}} q^{(i)}_{jk}}{{h}_j(\bm{x})} \\&  \leq  \frac{\hat{\theta}_j \pi_j+\left(1-\hat{\theta}_j\right)\displaystyle\sum_{u_i\in \mathcal{U}} \sum_{v_k\in\mathcal{V}} q^{(i)}_{jk}}{\hat{h}_j(\bm{x})} .
                    \end{aligned}
          \end{equation}
          }
          \vspace{-1mm}
          
    \noindent A similar proof holds for the rest of the constraints.
         
        \textbf{2) The equality of the constraints in~\eqref{eq:formulationAxiliaryRelaxed} to the constraints in~\eqref{eq:formulationAxiliary} upon convergence; mathematically, for every inequality in the original problem~\eqref{eq:formulationAxiliary} in the form of $v(\bm{x})\leq1$ and its approximated version $\hat{v}(\bm{x})\leq1$ in~\eqref{eq:formulationAxiliaryRelaxed}, we should have $v(\bm{x}^{[m]})= \hat{v}(\bm{x}^{[m]})$}. As an example, we prove that this holds between $\mathbf{\tilde{C}1}$ and $\mathbf{\hat{C}1}$ in~\eqref{eq:longPrp1}. Note that since the numerators of the two constraints are the same, examining the equality of the denominators is sufficient.  The proof for the rest of the constraints is similar and omitted for brevity.

  \textbf{3) The KKT conditions of~\eqref{eq:formulationAxiliary} should be satisfied after
the series of approximations converges in~\eqref{eq:formulationAxiliaryRelaxed}; mathematically, for every inequality in the original problem~\eqref{eq:formulationAxiliary} in the form of $v(\bm{x})\leq1$ and its approximated version $\hat{v}(\bm{x})\leq1$ in~\eqref{eq:formulationAxiliaryRelaxed}, we should have $\triangledown  v(\bm{x}^{[m]})= \triangledown \hat{v}(\bm{x}^{[m]})$.} In~\eqref{eq:longPrp2}, we prove this for a general approximation of the ratio of two posynomials, where $\frac{j(\bm{x})}{g(\bm{x})}$ is approximated by $\frac{j(\bm{x})}{\hat{g}(\bm{x})}$, and $g$ and $\hat{g}$ have the format given in \eqref{eq:approxPosMon}. The proof for the rest of partial derivatives, and thus the gradient, is similar.

 The poof of the proposition is the direct consequence of combining the above result with Fact~\ref{fact:2} and Fact~\ref{fact:1}.
      \vspace{-2mm}
\bibliographystyle{IEEEtran}
\bibliography{RandomWalkUAV}
\vspace{-13mm}
\begin{IEEEbiography}[{\includegraphics[width=1.0in,height=1.15in,clip]{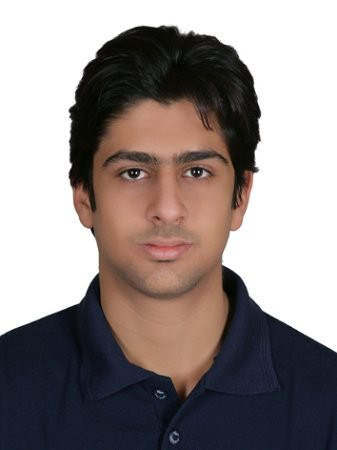}}]{Seyyedali Hosseinalipour} (S'18-M'20)
received B.S. degree from Amirkabir University of Technology in 2015, and M.S. and Ph.D. degree from NC State University in 2017 and 2020, respectively, all in electrical engineering. He received 2020 ECE doctoral scholar of the year award at NC State. He is currently a postdoctoral researcher at Purdue University. His research interests mainly include analysis of modern wireless networks and communication systems.
\end{IEEEbiography}
\vspace{-14mm}
\begin{IEEEbiography}[{\includegraphics[width=1.0in,height=1.15in,clip]{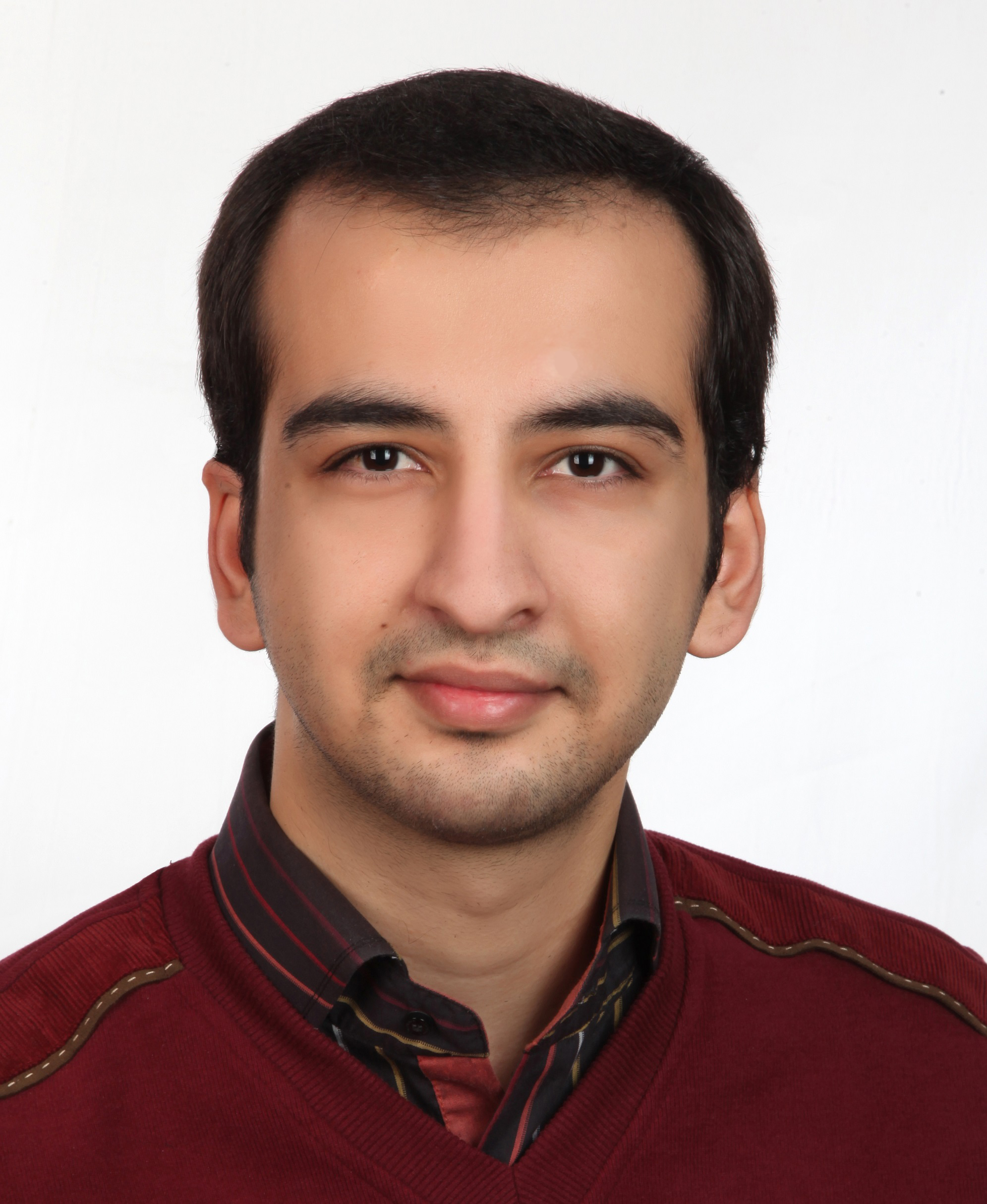}}]{Ali Rahmati} (S'15) received the B.Sc. degree in Electrical Engineering from the Ferdowsi University of Mashhad, Mashhad, Iran, and the M.S. degree in Electrical Engineering from University of Tehran, Tehran, Iran. He is currently pursuing the Ph.D. in the Department of Electrical and Computer Engineering, North Carolina State University, Raleigh, NC. His research interests mainly include applications of game theory, optimization and machine learning in wireless communication networks. 
\end{IEEEbiography}
	\vspace{-15mm}
\begin{IEEEbiography}[{\includegraphics[width=1.12in,height=1.25in,clip]{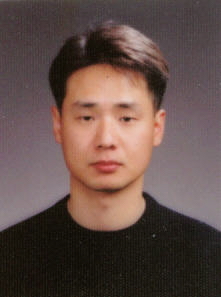}}]{Do Young Eun}  (M’03–SM’15) received his B.S. and M.S. degree in Electrical Engineering from Korea Advanced Institute of Science and Technology (KAIST), Taejon, Korea, in 1995 and 1997, respectively, and Ph.D. degree from Purdue University, West Lafayette, IN, in 2003. Since August 2003, he has been with the Department of Electrical and Computer Engineering at North Carolina State University, Raleigh, NC, where he is now a professor. His research interests include network modeling and performance analysis, mobile ad-hoc/sensor networks, mobility modeling, and randomized algorithms for large (social) networks. He has been a member of Technical Program Committee of various conferences including IEEE INFOCOM, ICC, Globecom, ACM MobiHoc, and ACM Sigmetrics. He is currently on the editorial board of IEEE/ACM Transactions on Networking and Computer Communications Journal, and was TPC co-chair of WASA'11. He received the Best Paper Awards in the IEEE ICCCN 2005, IEEE IPCCC 2006, and IEEE NetSciCom 2015, and the National Science Foundation CAREER Award 2006. He supervised and co-authored a paper that received the Best Student Paper Award in ACM MobiCom 2007.  
\end{IEEEbiography}
	\vspace{-12mm}
\begin{IEEEbiography}[{\includegraphics[width=1.15in,height=1.15in,clip]{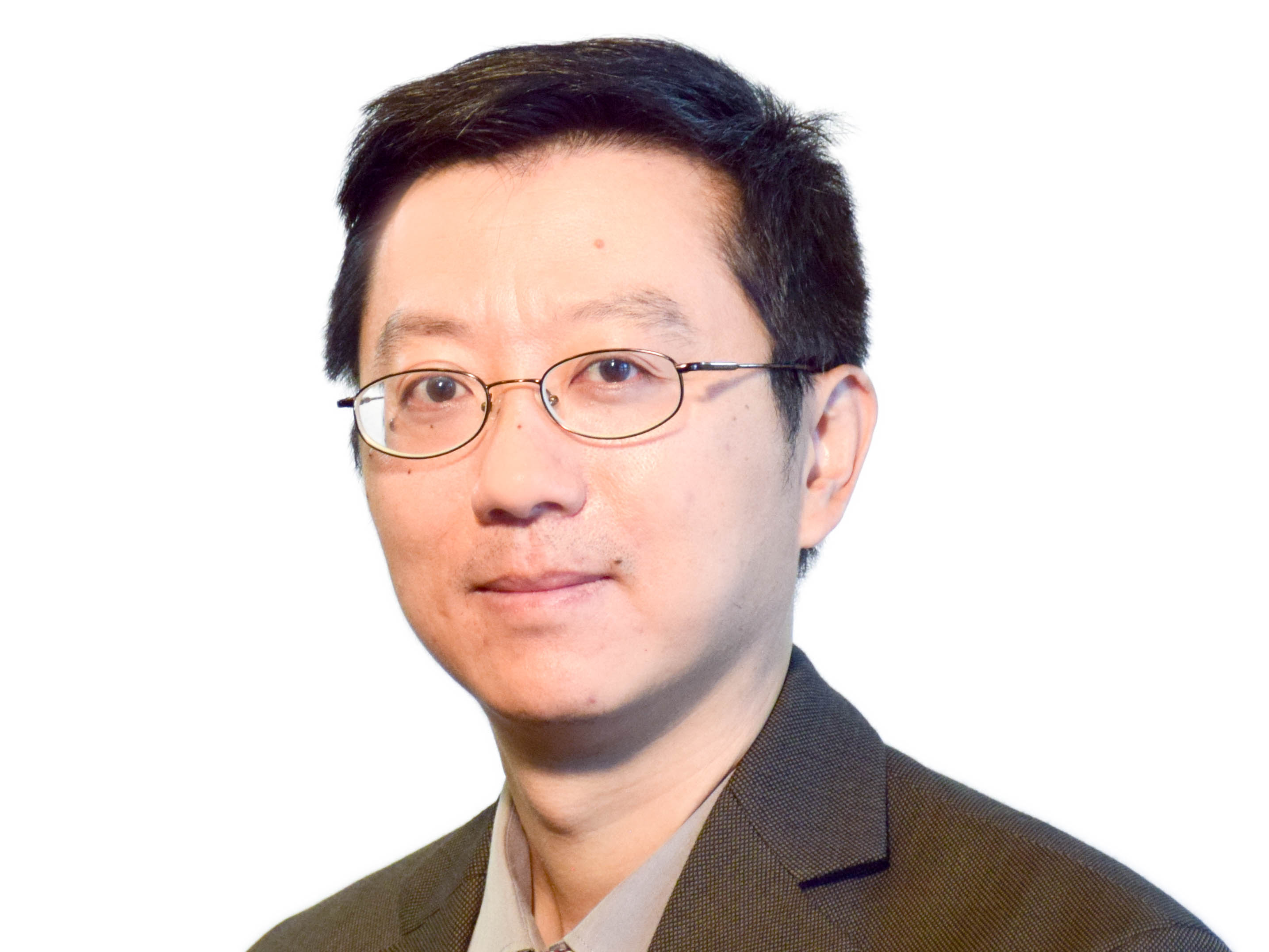}}]{Huaiyu Dai} (F’17)
  received the B.E. and M.S. degrees in electrical engineering from Tsinghua University, Beijing, China, in 1996 and 1998, respectively, and the Ph.D. degree in electrical engineering from Princeton University, Princeton, NJ in 2002. 

He was with Bell Labs, Lucent Technologies, Holmdel, NJ, in summer 2000, and with AT\&T Labs-Research, Middletown, NJ, in summer 2001. He is currently a Professor of Electrical and Computer Engineering with NC State University, Raleigh, holding the title of University Faculty Scholar. His research interests are in the general areas of communications, signal processing, networking, and computing. His current research focuses on machine learning and artificial intelligence for communications and networking, multilayer and interdependent networks, dynamic spectrum access and sharing, as well as security and privacy issues in the above systems.

He has served as an editor of IEEE Transactions on Communications, IEEE Transactions on Signal Processing, and IEEE Transactions on Wireless Communications. Currently he is an Area Editor in charge of wireless communications for IEEE Transactions on Communications, and a member of the Executive Editorial Committee for IEEE Transactions on Wireless Communications. He was a co-recipient of best paper awards at 2010 IEEE International Conference on Mobile Ad-hoc and Sensor Systems (MASS 2010), 2016 IEEE INFOCOM BIGSECURITY Workshop, and 2017 IEEE International Conference on Communications (ICC 2017).
\end{IEEEbiography}
    \end{document}